%% file: main.tex
\documentclass[acmsmall,screen]{acmart}

\usepackage{preamble}
\input{section/1-introduction-kl.tex}
\input{section/2-background-kl.tex}
\input{section/3-analysis-kl.tex}
\input{section/4-solution-kl.tex}
\knowledgescope{document}

\setcopyright{acmlicensed}
\copyrightyear{2026}
\acmYear{2026}
\acmDOI{XXXXXXX.XXXXXXX}
\acmConference[ICFP '26]{The ACM SIGPLAN International Conference on Functional Programming}{August 24--29, 2026}{Indianapolis, United States}
\acmISBN{XXX-X-XXXX-XXXX-X}
\keywords{Property-Based Testing, Denotational Semantics, Probabilistic Programming, Program Equivalence}

\ccsdesc[500]{Theory of computation~Probabilistic computation}
\ccsdesc[500]{Theory of computation~Denotational semantics}
\ccsdesc[500]{Theory of computation~Program reasoning}

\begin{document}

\title{Compositional Generator Equivalence \ifextended(Extended Version)\fi}
\author{Anthony Vandikas}
\email{anthony.vandikas@mail.utoronto.ca}
\author{Kiarash Sotoudeh}
\email{kiarash@cs.toronto.edu}
\author{Marsha Chechik}
\email{chechik@cs.utoronto.edu}
\affiliation{%
 \institution{University of Toronto}
 \city{Toronto}
 \state{Ontario}
 \country{Canada}
}

\begin{abstract}
  Property-based testing (PBT) is a powerful technique for software verification that relies on random input generators and ``shrinking'' processes to find and minimize counterexamples to executable specifications called properties. While optimizing these generators is crucial for testing efficiency, formally justifying such optimizations is currently difficult because existing languages lack a compositional semantics that is coarse-grained enough for high-level reasoning.

  In this paper, we first provide a formal account of the syntax and semantics of \hedgehog{}, a popular PBT framework. We demonstrate that \hedgehog{}'s distribution semantics --- which models how users typically reason about generators --- is non-compositional. Furthermore, we prove that any sound and complete compositional semantics for \hedgehog{} must necessarily be equivalent to its sampling semantics, which is too fine-grained to justify common program optimizations.

  To resolve this dilemma, we introduce \hedgehog*{}, a restricted version of the language based on the arrow calculus, and prove that \hedgehog*{} possesses a compositional distribution semantics. We evaluate \hedgehog*{} through a Haskell implementation and show that it remains expressive enough to capture generators of practical interest, while providing the formal foundation needed for compositional generator equivalence proofs.
\end{abstract}

\maketitle

\input{section/1-introduction.tex}
\input{section/2-background.tex}
\input{section/3-analysis.tex}
\input{section/4-solution.tex}
\input{section/5-evaluation.tex}
\input{section/6-related-work.tex}
\input{section/7-conclusion.tex}

\section*{Data Availability Statement}
The source code for our implementation, examples, and experiments is available at \url{https://github.com/hedgehog-arr/hedgehog-arr}. \ifextended\else The proofs for our theorems can be found in the extended version at
\url{https://github.com/hedgehog-arr/hedgehog-arr/blob/master/extended.pdf}.\fi

\end{document}

%% file: section/1-introduction-kl.tex
\knowledge{scope=introduction,notion}
| Property-based testing
| property-based testing
| PBT

\knowledge{scope=introduction,notion}
| PBT framework

\knowledge{scope=introduction,notion}
| property

\knowledge{scope=introduction,notion}
| system under test
| SUT

\knowledge{scope=introduction,notion}
| test case
| test cases

\knowledge{scope=introduction,notion}
| counterexample
| counterexamples

\knowledge{scope=introduction,notion}
| generator
| generators
| Generators

\knowledge{scope=introduction,notion}
| oracle
| oracles

\knowledge{scope=introduction,notion}
| shrinker
| shrinkers
| Shrinkers

\knowledge{scope=introduction,notion}
| testing stage
| testing

\knowledge{scope=introduction,notion}
| shrinking stage
| shrinking
| Shrinking

\knowledge{scope=introduction,notion}
| shrink
| shrinks

\knowledge{scope=introduction,notion}
| manually
| manual shrinking

\knowledge{scope=introduction,notion}
| generator-based
| generator-based shrinking

\knowledge{scope=introduction,notion}
| user-controlled
| user-controlled shrinking

\knowledge{scope=introduction,notion}
| optimization
| optimizations
| optimized

\knowledge{scope=introduction,notion}
| semantically equivalent
| semantic equivalence
| Semantic equivalence

\knowledge{scope=introduction,notion}
| sampling semantics

\knowledge{scope=introduction,notion}
| fine-grained
| coarse-grained

\knowledge{scope=introduction,notion}
| distribution semantics

\knowledge{scope=introduction,notion}
| compositional

\knowledge{scope=introduction,notion}
| observational equivalence

%% file: section/2-background-kl.tex
\knowledgenewunicodechar{ℕ}{\cmdkl{\mathbb{N}}}
\knowledge{synonym}
| natural number@background
| natural numbers@background

\knowledgenewrobustcmd{\ennreal}{\cmdkl{\overline{\mathbb{R}}_{≥ 0}}}
\knowledge{synonym}
| extended non-negative real number@background
| extended non-negative real numbers@background

\knowledgenewunicodechar{∪}[\mathbin]{\cmdkl{\cup}}
\knowledge{union}{synonym}

\knowledgenewunicodechar{∩}[\mathbin]{\cmdkl{\cap}}
\knowledge{intersection}{synonym}

\knowledgenewunicodechar{∅}{\cmdkl{\emptyset}}
\knowledge{empty set}{synonym}

\knowledgenewunicodechar{𝓟}{\cmdkl{\mathcal{P}}}
\knowledge{powerset}{synonym}

\knowledgenewrobustcmd{\Sample}{\cmdkl{\mathbb{S}}}
\knowledge{synonym}
| sample@background
| samples@background
| sample space@background

\knowledgenewunicodechar{∏}[\mathop]{\cmdkl{\prod}}
\knowledge{synonym}
| (cartesian) product@background
| cartesian product@background

\knowledgenewunicodechar{×}[\mathbin]{\cmdkl{\times}}
\knowledge{synonym}
| $n$-ary product@background

\knowledgenewrobustcmd{\one}{\cmdkl{\mathbf{1}}}
\knowledge{synonym}
| nullary product@background

\knowledgenewrobustcmd{\sing}{\cmdkl{()}}

\knowledgenewwrapped{\fprod}[\mathbin]{\cmdkl{\times}}

\knowledgenewsymb{\assoc}{\cmdkl{\mathrm{assoc}}}

\knowledgenewwrapped{\join}[\mathpunct]{\cmdkl{,}}
\knowledge{synonym}
| join@background
| joining@background

\knowledgenewunicodechar{→}[\mathbin]{\cmdkl{\rightarrow}}
\knowledge{synonym}
| function@background
| functions@background

\knowledgenewunicodechar{↦}[\mathbin]{\cmdkl{\mapsto}}
\knowledge{synonym}
| (anonymous) function@background

\knowledgenewwrapped{\blank}[\mathbin]{\cmdkl{-}}
\knowledge{synonym}
| blanks@background

\knowledgenewsymb{\id}{\cmdkl{\mathrm{id}}}
\knowledge{synonym}
| identity function@background

\knowledgenewunicodechar{∘}[\mathbin]{\cmdkl{\circ}}
\knowledge{synonym}
| function composition@background
| composition@background

\knowledgenewsymb{\curry}{\cmdkl{\mathrm{curry}}}

\knowledgenewunicodechar{∑}[\mathop]{\cmdkl{\sum}}
\knowledge{synonym}
| coproduct@background
| disjoint sum@background

\knowledgenewsymb{\inj}{\cmdkl{\iota}}

\newrobustcmd{\elim}[1]{\app{\kl[coproduct]{[} #1 \kl[coproduct]{]}}}
\knowledge{\elim}{notion}

\knowledgenewwrapped{\bisum}[\mathbin]{\cmdkl{+}}
\knowledge{synonym}
| $n$-ary coproduct@background

\knowledgenewrobustcmd{\zero}{\cmdkl{\mathbf{0}}}
\knowledge{synonym}
| nullary coproduct@background

\knowledgenewsymb{\absurd}{\cmdkl{!}}

\knowledgenewrobustcmd{\B}{\cmdkl{\mathbb{B}}}
\knowledge{synonym}
| boolean@background
| booleans@background

\knowledgenewrobustcmd{\true}{\cmdkl{\mathrm{true}}}
\knowledgenewrobustcmd{\false}{\cmdkl{\mathrm{false}}}

\knowledgenewsymb{\distrib}{\cmdkl{\mathrm{distrib}}}

\knowledgenewrobustcmd{\iverson}[1]{\mathopen{\cmdkl{[}}#1\mathclose{\cmdkl{]}}}
\knowledge{Iverson brackets}{synonym}

\knowledgenewrobustcmd{\Meas}{\cmdkl{\mathbf{Meas}}}
\knowledge{synonym}
| measurable space@background

\KnowledgeNewDocumentCommand{\mset}{E{_}{{}}}{\cmdkl{\mathcal{F}_{#1}}}
\knowledge{$σ$-algebra}{synonym}

\knowledgenewwrapped{\mfun}[\mathbin]{\cmdkl{\rightarrow_\mathbf{M}}}
\knowledge{synonym}
| measurable@background
| measurable function@background
| measurable functions@background

\knowledgenewsymb{\Borel}{\cmdkl{\mathrm{Borel}}}

\knowledgenewrobustcmd{\Qbs}{\cmdkl{\textbf{Qbs}}}
\knowledge{synonym}
| quasi-Borel space@background
| quasi-Borel spaces@background

\knowledgenewsymb{\Elem}{\cmdkl{\textrm{Elem}}}
\knowledge{synonym}
| random variable@background
| random variables@background

\knowledge{notion}
| random function@background
| random functions@background

\knowledgenewwrapped{\qfun}[\mathbin]{\cmdkl{\rightarrow_\textbf{Q}}}
\knowledge{synonym}
| quasi-measurable@background
| quasi-measurable function@background
| quasi-measurable functions@background

\knowledgenewsymb{\map}{\cmdkl{\mathrm{map}}}
\knowledge{synonym}
| functor@background

\knowledge{notion}
| pre-monad@background
| pre-monads@background

\knowledgenewsymb{\bind}{\cmdkl{\mathrm{bind}}}
\knowledgenewsymb{\unit}{\cmdkl{\mathrm{unit}}}

\knowledge{notion}
| monad@background
| monads@background

\knowledge{notion}
| pre-arrow@background
| pre-arrows@background

\knowledgenewsymb{\arr}{\cmdkl{\mathrm{arr}}}
\knowledgenewwrapped{\acomp}[\mathbin]{\cmdkl{\circ}}
\knowledgenewsymb{\afirst}{\cmdkl{\mathrm{first}}}
\knowledgenewsymb{\aleft}{\cmdkl{\mathrm{left}}}

\knowledge{notion}
| arrow@background
| arrows@background

\knowledgenewsymb{\Dist}{\cmdkl{\mathrm{Dist}}}
\knowledge{synonym}
| distribution@background
| distribution monad@background
| measure@background
| measure monad@background

\KnowledgeNewDocumentCommand{\integral}{O{}mm}{\mathchoice
  {\mathop{\cmdkl{\int}}\displaylimits\IfNoValueF{#1}{_{\mathrlap{#1}\phantom{x}}}#2\,\cmdkl{\mathrm{d}}{#3}}
  {\mathop{\cmdkl{\int}}\displaylimits\IfNoValueF{#1}{_{#1}}#2\,\kl[integral]{\mathrm{d}}{#3}}
  {\mathop{\cmdkl{\int}}\displaylimits\IfNoValueF{#1}{_{#1}}#2\,\kl[integral]{\mathrm{d}}{#3}}
  {\mathop{\cmdkl{\int}}\displaylimits\IfNoValueF{#1}{_{#1}}#2\,\cmdkl{\mathrm{d}}{#3}}}
\knowledge{synonym}
| integral@background
| integration@background
| integration operator@background
| expectation@background
| expectation operator@background

\knowledge{notion}
| push-forward@background

\knowledgenewwrapped{\dmul}[\mathbin]{⋅}

\knowledgenewrobustcmd{\uniform}{\cmdkl{λ}}
\knowledge{synonym}
| uniform distribution@background

\knowledgenewsymb{\projl}{\cmdkl{πₗ}}
\knowledgenewsymb{\projr}{\cmdkl{πᵣ}}

\knowledgenewsymb{\List}{\cmdkl{\mathrm{List}}}
\knowledge{synonym}
| list@background
| lists@background

\knowledgenewunicodechar{ε}{\cmdkl{\epsilon}}
\knowledgenewunicodechar{∷}[\mathbin]{\cmdkl{\dblcolon}}
\knowledgenewunicodechar{⧺}[\mathbin]{\cmdkl{{+}\kern-0.7ex{+}}}

\knowledgenewsymb{\Tree}{\cmdkl{\mathrm{Tree}}}
\knowledge{synonym}
| (rose) tree@background
| (rose) trees@background
| rose tree@background
| rose trees@background
| tree@background
| trees@background

\knowledgenewsymb{\node}{\cmdkl{\mathrm{node}}}

%% file: section/3-analysis-kl.tex
\knowledgenewrobustcmd\Type{\cmdkl{\mathrm{Type}}}
\knowledge{synonym}
| types@hedgehog
| type@hedgehog

\knowledgenewrobustcmd{\Empty}{\cmdkl{\syntax{0}}}
\knowledge{synonym}
| empty type@hedgehog
| empty types@hedgehog

\knowledgenewrobustcmd{\Unit}{\cmdkl{\syntax{1}}}
\knowledge{synonym}
| unit type@hedgehog
| singleton type@hedgehog

\knowledgenewrobustcmd{\Sum}[2]{#1 \mathbin{\cmdkl{\syntaxop{+}}} #2}
\knowledge{synonym}
| sum type@hedgehog
| sum types@hedgehog

\knowledgenewrobustcmd{\Prod}[2]{#1 \mathbin{\cmdkl{\syntaxop{*}}} #2}
\knowledge{synonym}
| product type@hedgehog
| product types@hedgehog

\knowledgenewrobustcmd{\Fun}[2]{#1 \mathbin{\cmdkl{\syntax{\shortrightarrow}}} #2}
\knowledge{synonym}
| function type@hedgehog
| function types@hedgehog

\knowledgenewsymb{\Gen}{\cmdkl{\syntax{Gen}}}
\knowledge{synonym}
| generator@hedgehog
| generators@hedgehog
| generator type@hedgehog
| generator types@hedgehog

\knowledgenewrobustcmd\Term{\cmdkl{\mathrm{Term}}}
\knowledge{synonym}
| terms@hedgehog
| term@hedgehog

\knowledgenewsymb{\tabsurd}{\cmdkl{\syntax{absurd}}}
\knowledge{synonym}
| absurd@hedgehog

\knowledgenewrobustcmd{\tunit}{\cmdkl{\syntaxop{()}}}
\knowledge{synonym}
| singleton value@hedgehog

\knowledgenewsymb{\tinl}{\cmdkl{\syntax{inl}}}
\knowledge{synonym}
| left injection@hedgehog

\knowledgenewsymb{\tinr}{\cmdkl{\syntax{inr}}}
\knowledge{synonym}
| right injection@hedgehog

\knowledgenewrobustcmd{\tmatch}[5]{%
  \cmdkl{\keyword{match}}\;#1\;
  \cmdkl{\keyword{with}}\;
  \app{\cmdkl{\syntax{inl}}}{#2}\mathpunct{\cmdkl{\syntaxop{.}}}#3\mathbin{\cmdkl{\syntaxop{|}}}  \app{\cmdkl{\syntax{inr}}}{#4}\mathpunct{\cmdkl{\syntaxop{.}}}#5}
\WithSuffix\newcommand\tmatch*[5]{\withkl{\kl[\tmatch]}{%
  \begin{array}[t]{@{}l@{}}
    \cmdkl{\keyword{match}}\;#1\;\cmdkl{\keyword{with}}                                                        \\
    \mathbin{\phantom{\syntaxop{|}}} \app{\cmdkl{\syntax{inl}}}{#2}\mathpunct{\cmdkl{\syntaxop{.}}}#3 \\
    \mathbin{\cmdkl{\syntaxop{|}}}   \app{\cmdkl{\syntax{inr}}}{#4}\mathpunct{\cmdkl{\syntaxop{.}}}#5}
  \end{array}}
\newcommand{\tmatchkeyword}{\kl[\tmatch]{\keyword{match}}}
\knowledge{synonym}
| case analysis@hedgehog

\knowledgenewrobustcmd{\tpair}[1]{\cmdkl{\syntaxop{(}}#1\cmdkl{\syntaxop{)}}}
\knowledge{synonym}
| tuple@hedgehog
| tuples@hedgehog

\knowledgenewsymb{\tfst}{\cmdkl{\syntax{fst}}}
\knowledge{synonym}
| left projection@hedgehog
| left projections@hedgehog

\knowledgenewsymb{\tsnd}{\cmdkl{\syntax{snd}}}
\knowledge{synonym}
| right projection@hedgehog
| right projections@hedgehog

\knowledgenewrobustcmd\tlam[4][]{
  \cmdkl{\keyword{fun}}^{#1}\;
  #2 \mathbin{\cmdkl{\syntax{:}}}
  #3 \mathpunct{\cmdkl{\syntaxop{.}}}
  #4}
\WithSuffix\newcommand\tlam*[4][]{{
      \begin{array}[t]{@{}l@{}}
        \tlam[#1]{#2}{#3}{ \\ #4}
      \end{array}}}
\knowledge{synonym}
| anonymous function@hedgehog
| anonymous functions@hedgehog

\knowledge{notion}
| function application@hedgehog

\knowledgenewsymb{\treturn}{\cmdkl{\keyword{return}}}
\knowledge{synonym}
| effect-free computation@hedgehog
| effect-free computations@hedgehog

\knowledgenewrobustcmd\tlet[3]{
  \cmdkl{\keyword{let}}\;#1 \mathbin{\cmdkl{\syntaxop{\textleftarrow}}} #2\;
  \cmdkl{\keyword{in}}\;#3}
\newcommand\tletkeyword{\kl[\tlet]{\keyword{let}}}
\WithSuffix\newcommand\tlet*[4][t]{{%
      \begin{array}[#1]{@{}l@{}}
        \tlet{#2}{#3 \\\negthickspace}{
          #4}
      \end{array}}}
\WithSuffix\newcommand\tlet+[4][t]{{%
      \begin{array}[#1]{@{}l@{}}
        \tlet{#2}{#3}{ \\
          #4}
      \end{array}}}
\knowledge{synonym}
| sequential composition@hedgehog

\knowledgenewsymb{\tflip}{\cmdkl{\syntax{flip}}}
\knowledge{synonym}
| coin flip operation@hedgehog

\knowledgenewsymb{\tshrink}{\cmdkl{\syntax{shrink}}}
\knowledge{synonym}
| shrinking operation@hedgehog

\knowledgenewsymb{\Var}{\cmdkl{\mathrm{Var}}}
\knowledge{synonym}
| variable@hedgehog
| variables@hedgehog

\knowledgenewrobustcmd{\Bool}{\cmdkl{\syntax{Bool}}}
\knowledgenewrobustcmd{\ttrue}{\cmdkl{\syntax{true}}}
\knowledgenewrobustcmd{\tfalse}{\cmdkl{\syntax{false}}}
\knowledgenewsymb{\tnot}{\cmdkl{\syntax{not}}}
\knowledgenewwrapped{\teq}[\mathbin]{\cmdkl{\syntaxop{=}}}
\knowledgenewwrapped{\tand}[\mathbin]{\cmdkl{\syntaxop{and}}}
\knowledgenewwrapped{\tor}[\mathbin]{\cmdkl{\syntaxop{or}}}

\knowledgenewrobustcmd{\tif}[3]{
  \cmdkl{\keyword{if}}\;#1\;
  \cmdkl{\keyword{then}}\;#2\;
  \cmdkl{\keyword{else}}\;#3}
\WithSuffix\newcommand\tif*[3]{
  \begin{array}[t]{@{}l@{}}
    \tif{#1}{ \\
    \quad #2  \\\negthickspace}{\\
      \quad #3}
  \end{array}}
\WithSuffix\newcommand\tif+[3]{
  \begin{array}[t]{@{}l@{}}
    \tif{#1 \\\negthickspace}{
    #2      \\\negthickspace}{
      #3}
  \end{array}}
\newcommand{\tifkeyword}{\kl[\tif]{\keyword{if}}}

\knowledgenewrobustcmd{\tcoini}{\app{\cmdkl{\syntax{coin₁}}}}
\knowledgenewrobustcmd{\tcoinii}{\app{\cmdkl{\syntax{coin₂}}}}
\knowledgenewrobustcmd{\tcoiniii}{\app{\cmdkl{\syntax{coin₁'}}}}
\knowledgenewrobustcmd{\tcoinix}{\app{\cmdkl{\syntax{coin₁^{\arrow}}}}}
\knowledgenewrobustcmd{\tcoiniix}{\app{\cmdkl{\syntax{coin₂^{\arrow}}}}}

\knowledgenewrobustcmd\Env{\cmdkl{\mathrm{Env}}}
\knowledge{synonym}
| environment@hedgehog
| environments@hedgehog

\knowledgenewsymb{\Dom}{\cmdkl{\mathrm{Dom}}}
\knowledge{synonym}
| domain@hedgehog

\KnowledgeNewDocumentCommand\typed{omm}{%
  \IfNoValueF{#1}{#1 \mathrel{\cmdkl{⊢}}} #2 \mathrel{\cmdkl{:}} #3}
\knowledge{synonym}
| typing relation@hedgehog

\knowledgenewsymb{\flip}{\cmdkl{\mathrm{flip}}}
\knowledgenewsymb{\shrink}{\cmdkl{\mathrm{shrink}}}

\knowledge{notion}
| interpretation@hedgehog
| interpretations@hedgehog
| compositional@hedgehog

\knowledgenewrobustcmd{\tysem}[1]{\mathopen{\cmdkl{⟦}}#1\mathclose{\cmdkl{⟧}}}
\knowledge{synonym}
| type and environment semantics@hedgehog
| environment semantics@hedgehog
| type semantics@hedgehog
| variable assignment@hedgehog
| variable assignments@hedgehog

\NewDocumentCommand{\tmsem}{omo}{\app{\withkl{\kl[\tmsem]}{%
      \mathopen{\cmdkl{⟦}}%
      \IfNoValueF{#1}{#1 \mathrel{\cmdkl{⊢}}}%
      #2%
      \IfNoValueF{#3}{\mathrel{\cmdkl{:}} #3}
      \mathclose{\cmdkl{⟧}}}}}
\knowledge{notion}
| \tmsem
| term semantics@hedgehog

\knowledgenewsymb{\Sampling}{\cmdkl{\mathfrak{C}}}
\knowledge{synonym}
| sampling semantics@hedgehog
| sampling interpretation@hedgehog
| Sampling interpretation@hedgehog

\knowledge{notion}
| shrink tree@hedgehog
| shrink trees@hedgehog

\KnowledgeNewDocumentCommand{\semeq}{mommo}{
  \IfNoValueF{#2}{#2 \mathrel{\cmdkl{⊢}}}
  #3 \mathrel{\cmdkl{=}^{#1}} #4
  \IfNoValueF{#5}{\mathrel{\cmdkl{:}} #5}}
\knowledge{synonym}
| semantic equivalence@hedgehog
| semantically equivalent@hedgehog

\knowledgenewrobustcmd{\dist}[1]{\mathopen{\cmdkl{⦇}}#1\mathclose{\cmdkl{⦈}}}
\knowledge{synonym}
| distribution semantics@hedgehog
| distributions@hedgehog
| distribution@hedgehog

\knowledgenewsymb{\tcontext}{\cmdkl{\syntax{discrim}}} 
\knowledgenewsymb{\tcontextx}{\cmdkl{\syntax{discrim}^{\arrow}}} 

\knowledgenewrobustcmd{\Ctx}{\cmdkl{\mathrm{Ctx}}}
\knowledge{synonym}
| context@hedgehog
| contexts@hedgehog

\knowledgenewrobustcmd{\ctxtyped}[5]{
  #1 \mathrel{\cmdkl{:}}
  (#2 \mathrel{\cmdkl{⊢}} #3) \mathrel{\cmdkl{\longrightarrow}}
  (#4 \mathrel{\cmdkl{⊢}} #5)}

\knowledgenewrobustcmd{\hole}{\cmdkl{\mathrm{□}}}
\knowledge{synonym}
| hole@hedgehog
| holes@hedgehog

\knowledgenewrobustcmd{\subst}[2]{#1\mathopen{\cmdkl{[}}#2\mathclose{\cmdkl{]}}}
\knowledge{synonym}
| substitution@hedgehog

\KnowledgeNewDocumentCommand{\disteq}{ommo}{
  \IfNoValueF{#1}{#1 \mathrel{\cmdkl{⊢}}}
  #2 \mathrel{\cmdkl{≈}} #3
  \IfNoValueF{#4}{\mathrel{\cmdkl{:}} #4}}
\knowledge{synonym}
| contextual equivalence@hedgehog
| contextually equivalent@hedgehog

\knowledge{notion}
| sound@hedgehog
| soundness@hedgehog
| Soundness@hedgehog

\knowledge{notion}
| complete@hedgehog
| completeness@hedgehog
| Completeness@hedgehog

\knowledgenewsymb{\tdiscrim}{\cmdkl{\syntax{discrim'}}}

%% file: section/4-solution-kl.tex
\WithSuffix\knowledgenewrobustcmd\Type*{\cmdkl{\mathrm{Type}^{\arrow}}}
\knowledge{synonym}
| type@hedgehog*
| types@hedgehog*

\WithSuffix\knowledgenewrobustcmd\Term*{\cmdkl{\mathrm{Term}^{\arrow}}}
\knowledge{synonym}
| term@hedgehog*
| terms@hedgehog*

\WithSuffix\KnowledgeNewDocumentCommand\typed*{omm}{%
\IfNoValueF{#1}{#1 \mathrel{\cmdkl{⊢^{\arrow}}}}%
#2 \mathrel{\cmdkl{:}}%
#3}
\knowledge{synonym}
| typing relation@hedgehog*
| typing@hedgehog*

\newcommand{\semi}{\mathpunct{\cmdkl{;}}}
\KnowledgeNewDocumentCommand\ctyped{mmmm}{%
#1 \semi #2 \mathrel{\cmdkl{⊢^{\arrow}}} #3 \mathrel{\cmdkl{!}} #4}
\knowledge{synonym}
| command typing relation@hedgehog*

\WithSuffix\knowledgenewrobustcmd\Env*{\cmdkl{\mathrm{Env}^{\arrow}}}
\knowledge{synonym}
| environment@hedgehog*
| environments@hedgehog*

\knowledgenewrobustcmd{\GFun}[2]{#1 \mathbin{\cmdkl{\rightsquigarrow}} #2}
\knowledge{synonym}
| generator function type@hedgehog*
| generator function@hedgehog*
| generator functions@hedgehog*
| generator@hedgehog*
| generators@hedgehog*

\knowledge{notion}
| command term@hedgehog*
| command terms@hedgehog*

\knowledge{notion}
| interpretation@hedgehog*
| interpretations@hedgehog*
| compositional@hedgehog*

\WithSuffix\knowledgenewrobustcmd\tysem*[1]{\mathopen{\cmdkl{⟦}}#1\mathclose{\cmdkl{⟧}}}
\knowledge{synonym}
| type and environment semantics@hedgehog*
| environment semantics@hedgehog*
| type semantics@hedgehog*

\WithSuffix\NewDocumentCommand\tmsem*{omo}{\app{\withkl{\kl[\tmsem*]}{%
\mathopen{\cmdkl{⟦}}%
\IfNoValueF{#1}{#1 \mathrel{\cmdkl{⊢^{\arrow}}}}%
#2%
\IfNoValueF{#3}{\mathrel{\cmdkl{:}} #3}
\mathclose{\cmdkl{⟧}}}}}
\knowledge{notion}
| \tmsem*
| term semantics@hedgehog*

\WithSuffix\knowledgenewsymb\Sampling*{\cmdkl{\mathfrak{C}^{\arrow}}}
\knowledge{synonym}
| sampling semantics@hedgehog*
| sampling interpretation@hedgehog*
| Sampling interpretation@hedgehog*

\WithSuffix\KnowledgeNewDocumentCommand\semeq*{mommo}{%
\IfNoValueF{#2}{#2 \mathrel{\cmdkl{⊢^{\arrow}}}}%
#3 \mathrel{\cmdkl{=}^{#1}} #4%
\IfNoValueF{#5}{\mathrel{\cmdkl{:}} #5}}
\knowledge{synonym}
| semantic equivalence@hedgehog*
| semantically equivalent@hedgehog*

\WithSuffix\knowledgenewrobustcmd\dist*[1]{\cmdkl{⦇}#1\cmdkl{⦈^{\arrow}}}
\knowledge{synonym}
| distribution semantics@hedgehog*
| distributions@hedgehog*
| distribution@hedgehog*

\WithSuffix\knowledgenewrobustcmd\Ctx*{\cmdkl{\mathrm{Ctx^{\arrow}}}}
\knowledge{synonym}
| context@hedgehog*
| contexts@hedgehog*

\WithSuffix\knowledgenewrobustcmd\ctxtyped*[5]{
  #1 \mathrel{\cmdkl{:}}
  (#2 \mathrel{\cmdkl{⊢^{\arrow}}} #3) \mathrel{\cmdkl{\longrightarrow}}
  (#4 \mathrel{\cmdkl{⊢^{\arrow}}} #5)}

\knowledge{notion}
| translation@hedgehog*
| \tytrans
| \tmtrans
| \envtrans
\newcommand{\tytrans}[1]{\mathopen{\kl[\tytrans]{⟨}}#1\mathclose{\kl[\tytrans]{⟩_\mathrm{ty}}}}
\newcommand{\tmtrans}[1]{\mathopen{\kl[\tmtrans]{⟨}}#1\mathclose{\kl[\tmtrans]{⟩_\mathrm{tm}}}}
\newcommand{\envtrans}[1]{\mathopen{\kl[\envtrans]{⟨}}#1\mathclose{\kl[\envtrans]{⟩_\mathrm{env}}}}

\knowledgenewsymb{\Distrib}{\cmdkl{\mathfrak{D}}}
\knowledge{synonym}
| distribution interpretation@hedgehog*

%% file: section/1-introduction.tex
\begin{scope}
\knowledgeimport{introduction}
\knowledgeimport{background}

\section{Introduction} \label{sec:introduction}

\begin{outline}
  \textbf{Property-Based Testing}
  \begin{itemize}
    \item Property
    \item Test Cases
    \item Counterexample
  \end{itemize}
\end{outline}
""Property-based testing"" (PBT) is a powerful tool for identifying software bugs. Users provide an executable specification --- or ""property"" --- to a ""PBT framework"", which automatically checks the specification against a ""system under test"" ("SUT") over a wide range of randomly generated inputs (or ""test cases""). If the framework finds a ""counterexample"" --- a "test case" that causes the "SUT" to violate the specification --- then the framework computes and reports a minimal "counterexample".

\begin{outline}
  \textbf{Property}
  \begin{itemize}
    \item Generators
    \item Oracle
    \item Shrinkers
  \end{itemize}
\end{outline}
A "property" consists of three components:
\begin{enumerate*}
  \item ""generators"", which compute random "test cases",
  \item an ""oracle"", which decides if a test case is a counterexample, and
  \item ""shrinkers"", which compute a set of ``smaller'' "test cases" from an existing "test case".
\end{enumerate*}
A "property"'s components are integral to the operation of a "PBT framework", which consists of two stages: "testing" and "shrinking". The ""testing stage"" is responsible for identifying a "counterexample", and operates using the "generators" and the "oracle" of the "property". The ""shrinking stage"" performs "counterexample" generation and minimization, and operates using the "shrinkers" and the "oracle".


\begin{outline}
  \textbf{Testing Stage}
  \begin{itemize}
    \item Stopping Conditions
  \end{itemize}
\end{outline}
The "testing stage" proceeds by constructing a random "test case" using the "generators" and then checking if the "test case" is a "counterexample" using the "oracle". If the "test case" is not a "counterexample", the "testing" process restarts with a different "test case"; otherwise, the "counterexample" is passed along to the "shrinking stage". In lieu of a "counterexample", the "testing stage" can be terminated
\begin{enumerate*}
  \item after a user-specified number of "test cases" have been generated,
  \item after a user-specified timeout, or
  \item manually.
\end{enumerate*}

\begin{outline}
  \textbf{Shrinking Stage}
  \begin{itemize}
    \item Stopping Conditions
  \end{itemize}
\end{outline}
The "shrinking stage" is essential for the usability of "property-based testing". Without it, "counterexamples" can be arbitrarily large, which complicates debugging. "Shrinking" operates as follows:
\begin{enumerate*}
  \item starting from the previously discovered "counterexample", it computes a set of ``smaller'' "test cases", called ""shrinks"", and then
  \item the "oracle" determines whether any of the "shrinks" are "counterexamples". If so, the "shrinking" process restarts using the smaller "counterexample"; otherwise, the current "counterexample" is deemed minimal and reported to the user.
\end{enumerate*}
Like the "testing stage", the "shrinking stage" can be terminated
\begin{enumerate*}
  \item after a user-specified number of "shrinks" have been generated,
  \item after a user-specified timeout, or
  \item manually.
\end{enumerate*}

\begin{outline}
  \textbf{Generator Specification}
\end{outline}
"Generators" are specified using an embedded domain-specific language (eDSL) provided by the "PBT framework". These eDSLs are essentially probabilistic programming languages, and hence users view generators as probability distributions. For example, consider the program given in \Cref{fig:inequality-example}, written using the \hedgehog{} framework \citep{hedgehog} for \haskell{}. Concretely, the \texttt{coin1} generator computes two random booleans, \texttt{x} and \texttt{y} (lines 3 and 4), by invoking a sub-generator \texttt{bool\textunderscore} which performs a fair coin flip. The \texttt{coin1} generator then checks whether \texttt{x} and \texttt{y} are equal (line 5). Abstractly, \texttt{coin1} represents a probability distribution that assigns \texttt{0.5} to both \texttt{True} and \texttt{False}.

\begin{figure}[t]
  \centering
  \includegraphics[scale=0.8]{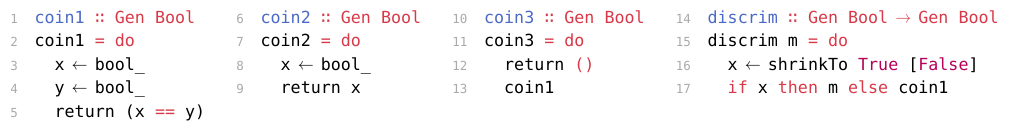}
    \vspace{-0.15in}
  \caption[Example generators and distinguishing context.]{Example generators and distinguishing context. The generator \texttt{shrinkTo\footnotemark\ True\ [False]} produces \texttt{True} during testing, but can shrink to \texttt{False} during shrinking.}
  \Description{Example generators and distinguishing context}
  \label{fig:inequality-example}
  \vspace{-0.25in}
\end{figure}
\footnotetext{The function \texttt{shrinkTo $\dblcolon$ a $\to$ [a] $\to$ Gen a} does not exist verbatim in Hedgehog, but is straightforward to implement using the existing combinator \texttt{shrink $\dblcolon$ (a $\to$ [a]) $\to$ Gen a $\to$ Gen a}.}

\begin{outline}
  \textbf{Shrinker Specification}
  \begin{itemize}
    \item Manual Shrinking
    \item Generator-Based Shrinking
    \item User-Controlled Shrinking
  \end{itemize}
\end{outline}
"Shrinkers" are also specified using an eDSL provided by the "PBT framework". The \quickcheck{} \citep{quickcheck} framework which popularized PBT employs ""manual shrinking"", where "shrinkers" and "generators" are specified separately. "Shrinkers" are surprisingly difficult to write "manually", thus subsequent work such as \hedgehog{} \citep{hedgehog} employs ""generator-based shrinking"", where "shrinkers" are derived from annotated "generator" specifications. Under "generator-based shrinking", "generators" represent distributions of labelled trees, where labels lower in the tree are considered ``smaller''. While most languages with "generator-based shrinking" allow users to influence "shrinker" behaviour \citep{hedgehog,falsify}, others --- such as \hypothesis{} \citep{hypothesis} --- rely solely upon built-in heuristics. This lack of ""user-controlled shrinking"" saves users from having to think about how to shrink "test cases", but inhibits correct and efficient "shrinking" since the correct notion of ``smaller'' "test case" is application-dependent. Furthermore, the heuristics used by these frameworks are unpredictable and can change when the framework is updated \citep{falsify}. {Thus, this work focuses on "user-controlled" "generator-based shrinking".}

\begin{outline}
  \textbf{Generator Efficiency}
  \begin{itemize}
    \item Efficiency
    \item Optimization
    \item Semantic Equivalence
  \end{itemize}
\end{outline}
Efficient "generators" and "shrinkers" are essential for the effectiveness of "property-based testing". When timeouts are used, slow "generators" reduce the likelihood that the "testing stage" discovers a "counterexample", and slow "shrinkers" reduce the likelihood that the "shrinking stage" completely minimizes "counterexamples". Thus it is important that generators and shrinkers are "optimized", either by hand or (preferably) automatically.

By ""optimization"", we mean replacing an existing (slow) program with a (faster) ""semantically equivalent"" one. Thus, in order to perform "optimizations", we must formally justify their validity via a "semantic equivalence" proof. "Semantic equivalence" in languages with generator-based shrinking is surprisingly subtle. For example, consider the generator \texttt{coin2} (\Cref{fig:inequality-example}, line 6). As \texttt{coin2} has the same distribution as \texttt{coin1}, we can ostensibly optimize a program using \texttt{coin1} by replacing \texttt{coin1} with \texttt{coin2}. However, this is generally unsound, as \texttt{coin1} and \texttt{coin2} produce observably different behaviour in combination with the function \texttt{discrim} (line 14): \texttt{discrim\ coin1} produces the trees given in \Cref{fig:expected-trees}, while \texttt{discrim\ coin2} produces the trees given in both \Cref{fig:expected-trees,fig:additional-trees}. We explain the reason for this in \Cref{ex:coin-distribution}.


\begin{outline}
  \textbf{Problem and Existing Solutions}
  \begin{itemize}
    \item Existing Probabilistic Calculi
    \item Existing Probabilistic Non-Determinism Calculi
    \item Implementation-Level Reasoning
          \begin{itemize}
            \item Sampling Semantics
            \item Sample Space
          \end{itemize}
    \item Distribution Reasoning
          \begin{itemize}
            \item Distribution Semantics
            \item Compositionality
          \end{itemize}
  \end{itemize}
\end{outline}
Unfortunately, equivalence proofs are currently infeasible as there is currently no way to formally and effectively reason about the behaviour of generators in languages with "user-controlled" "generator-based shrinking". Existing possibilities include:

\emph{(1) Utilize existing probabilistic calculi.} Since "generator" specification languages are effectively probabilistic programming languages, we can use existing probabilistic calculi to reason about "generator" behaviour. However, existing probabilistic calculi \citep{probabilistic-λ-calculus-1,probabilistic-λ-calculus-2,probabilistic-λ-calculus-3} do not model "shrinking" behaviour.

\emph{(2) Adapt existing calculi for probabilistic non-determinism.} "Shrinking" can be viewed as a form of non-determinism, which suggests that we can adapt existing calculi which model combinations of probabilistic and non-deterministic effects for our purposes. However, languages which feature "generator-based shrinking" violate common equational properties --- the \emph{monad laws}~\citep{monads} --- which are present in these calculi and also intuitively expected by users. For example, the monad laws imply that the generators \texttt{coin1} and \texttt{coin3} (\Cref{fig:inequality-example}, line 10) are equivalent. Intuitively, this is expected because \texttt{return ()} is a ``no-op'' and can be safely removed. Nonetheless, \texttt{coin1} and \texttt{coin3} (like \texttt{coin1} and \texttt{coin2}) are distinguished by \texttt{discrim}, as \texttt{discrim coin3} returns the same set of trees as \texttt{discrim coin2} (\Cref{fig:expected-trees,fig:additional-trees}).

\emph{(3) Implementation-level reasoning.} We can reason about languages directly in terms of their implementation (or ""sampling semantics""). Existing languages such as \hedgehog{} \citep{hedgehog} represent a generator as a \reintro{random variable}, i.e., a deterministic function on some \reintro{sample space}. However, "sampling semantics" is too ""fine-grained"", not preserved even by simple changes.
For example,  \texttt{coin1}, \texttt{coin2}, and \texttt{coin3}  (\Cref{fig:inequality-example}) are all semantically distinct under \hedgehog{}'s "sampling semantics" because they interpret their "samples" differently. In contrast, an appropriate \reintro{coarse-grained} semantics should allow \texttt{coin1}, \texttt{coin2}, and \texttt{coin3} to be identified, at least in some contexts.

\emph{(4) Distribution-level reasoning.} The effectiveness of property-based testing depends on the distribution of a generator, not the specifics of how "samples" are turned into "test cases". Since \hedgehog{} has a "sampling semantics" in terms of "random variables", one can derive an appropriate ""distribution semantics"". However, this "distribution semantics" is not ""compositional"": the distribution of a generator is not uniquely determined by the distributions of its components. For example, we cannot determine the distribution of \texttt{discrim\ m} (\Cref{fig:inequality-example}) simply by knowing the distribution of \texttt{m}: \texttt{coin1} and \texttt{coin2} have the same distribution, but \texttt{discrim\ coin1} and \texttt{discrim\ coin2} do not. A "compositional" semantics is important for reasoning about the behaviour of larger programs.

\begin{figure}[t]
  \begin{subfigure}[t]{0.25\linewidth}
    \centering
    \includegraphics[scale=0.8]{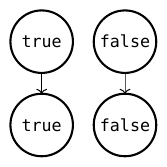}
    \caption{Expected Trees}
    \Description{Expected Trees}
    \label{fig:expected-trees}
  \end{subfigure}
  \begin{subfigure}[t]{0.25\linewidth}
    \centering
    \includegraphics[scale=0.8]{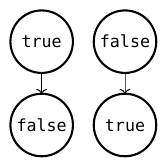}
    \caption{Additional Trees}
    \Description{Additional Trees}
    \label{fig:additional-trees}
  \end{subfigure}
  \vspace{-0.1in}
  \caption{Trees produced by \texttt{discrim} (see \Cref{fig:inequality-example}).}
  \vspace{-0.25in}
\end{figure}

\paragraph{Contributions.}
In this paper, we address the problem of formally reasoning about "generator" behaviour in languages with "user-controlled" "generator-based shrinking", for the purpose of justifying optimizations. We focus on \hedgehog{} \citep{hedgehog}, although our results are easily transferred to other languages. The examples given in \Cref{fig:inequality-example} demonstrate that a "compositional" semantics must be somewhat sensitive to a generator's internal structure.In fact, we show \hedgehog{}'s \textbf{only} "compositional" semantics is its "sampling semantics"; \hedgehog{} generators can (at best) be interpreted as deterministic programs with no opportunity to optimize or reason about their probabilistic behaviour. In light of this result, we define a restricted version of \hedgehog{} based on the arrow calculus \citep{arrow-calculus} called \hedgehog*{}, and prove that its "distribution semantics" is "compositional" and validates many common optimizations. We evaluate the expressiveness of \hedgehog*{} relative to \hedgehog{} in practice and demonstrate \hedgehog*{}'s utility for equivalence proofs on a set of examples.  Specifically,
\begin{enumerate*}
  \item we give a formal account of the \hedgehog{} language's syntax and its "sampling semantics", where "generators" represent "random variables" on an appropriate "sample space",
  \item we also define \hedgehog{}'s "distribution semantics", where generators represent probability distributions, and show that it is not "compositional",
  \item \hedgehog{}'s "distribution semantics" induces an ""observational equivalence"" relation, and we show that the only "compositional" semantics that is sound and complete with respect to "observational equivalence" is the "sampling semantics",
  \item we define \hedgehog*{}, a version of \hedgehog{} based on the arrow calculus \citep{arrow-calculus}, and show that the "distribution semantics" for this version of \hedgehog{} is "compositional", and
  \item we evaluate the expressiveness of \hedgehog*{}, its effect on program size, and its utility for program optimization.
\end{enumerate*}

\paragraph{Organization.} In \Cref{sec:background}, we review background material. In \Cref{sec:analysis}, we present our formalization and analysis of the \hedgehog{} language. In \Cref{sec:solution}, we present \hedgehog*{}. In \Cref{sec:evaluation}, we evaluate the practical utility of \hedgehog*{}. In \Cref{sec:related-work}, we describe and compare related work. In \Cref{sec:conclusion}, we conclude and discuss future work.

\end{scope}

%% file: section/2-background.tex
\begin{scope}
\knowledgeimport{background}

\section{Background} \label{sec:background}

In this section, we fix some basic notation (\Cref{sec:background/notation}) and then review the requisite probability theory (\Cref{sec:background/probability-theory}). We work within the framework of "quasi-Borel spaces" \citep{quasi-borel-spaces}, which allows reasoning about probabilistic behaviour of programs with higher-order functions.

\subsection{Notation} \label{sec:background/notation}

\paragraph{Sets.}
\AP
The set of ""natural numbers"" is denoted by $ℕ$ and
the set of ""extended non-negative real numbers"" is denoted by $\ennreal$.
The ""union"" of two sets $A$ and $B$ is denoted by $A ∪ B$,
the ""intersection"" is denoted by $A ∩ B$,
the ""empty set"" is denoted by $∅$, and
the ""powerset"" of $A$ is denoted by $𝓟{A}$.
The ""sample space"" $\Sample$ is defined as the interval $[0, 1) ⊆ \ennreal$.

\paragraph{Products.}
\AP
The ""(cartesian) product"" of a family of sets $(Aᵢ)_{i ∈ I}$ is defined as $∏_{i ∈ I} Aᵢ = \{\,(aᵢ)_{i ∈ I} | ∀ i ∈ I, aᵢ ∈ Aᵢ\,\}$.
We define the \intro{$n$-ary product} as $A₁ × … × Aₙ = ∏ⁿ_{i = 1} Aᵢ$, and write $(a₁, …, aₙ)$ in place of $(aᵢ)ⁿ_{i = 1}$.
The \intro{nullary product} is denoted by $\one$, and its unique element is denoted by $\intro*\sing$.
The ""joining"" of two families $a = (aᵢ)_{i ∈ I}$ and $b = (bⱼ)_{j ∈ J}$ such that $I ∩ J = ∅$ is denoted by $a \join b$.

\paragraph{Functions.}
\AP
The set of ""functions"" from set $A$ to set $B$ is denoted by $A → B$.
Function application is denoted using juxtaposition ($\app{f}{x}$).
The notation $(a ↦ …)$ denotes an ""(anonymous) function"" with parameter $a$.
Functions can be denoted using ""blanks"", e.g., $(1 + \blank) = (n ↦ 1 + n)$.
The ""identity function"" on $A$ is denoted by $\id_A$ and
the ""composition"" of two functions $f : B → C$ and $g : A → B$ is denoted by $f ∘ g$.
We omit the subscripts on $\id$ when it can be inferred from the context.

\paragraph{Coproducts.}
\AP
The ""coproduct"" (or \reintro{disjoint sum}) of a family of sets $(Aᵢ)_{i ∈ I}$ is defined as $∑_{i ∈ I} Aᵢ = \{\,\intro\inj_i{x} | x ∈ Aᵢ\,\}$.
We define the \reintro{$n$-ary coproduct} as $A₁ \intro*\bisum … \reintro*\bisum Aₙ = ∑ⁿ_{i = 1} Aᵢ$.
Given a family of functions $f : ∏_{i ∈ I} (Aᵢ → B)$, we define the function $\elim{fᵢ}_{i ∈ I} : ∑_{i ∈ I} Aᵢ → B$ to satisfy $\elim{fᵢ}_{i ∈ I}(\inj_j{x}) = \app{fⱼ}{x}$ for all $j ∈ I$.
The ""nullary coproduct"" is denoted by $\zero$ ($= ∅$), and
is equipped with a (unique) function $(\intro*\absurd_A) : \zero → A$ for all sets $A$.
We omit the subscript on $(\absurd)$ when it can be inferred from the context.
The set of ""booleans"" is defined as $\B = \one \bisum \one$, where $\intro*\true = \inj_1{\sing}$ and $\intro*\false = \inj_2{\sing}$.

\paragraph{Miscellaneous.} We implicitly treat propositions as "booleans". We define common miscellaneous functions in \Cref{fig:miscellaneous-functions}. We omit subscripts on these functions when they can be inferred from the context.

\begin{figure}
  \begin{mathpar}
    \begin{array}{@{}l@{}}
      (\intro*\fprod) : (A → B) × (A' → B') → A × A' → B × B' \\
      \app(f \fprod g)(x, y) = (\app{f}{x}, \app{g}{y})
    \end{array}

    \begin{array}{@{}l@{}}
      \intro*\distrib : A × (B \bisum C) → (A × B) \bisum (A × C) \\
      \distrib(x, \inj_1{y}) = \inj_1(x, y)                       \\
      \distrib(x, \inj_2{y}) = \inj_2(x, y)
    \end{array}

    \begin{array}{@{}l@{}}
      \intro*\curry_{A,B,C} : (A × B → C) → (A → B → C) \\
      \curry_{A,B,C}{f}{a}{b} = \app{f}(a, b)
    \end{array}

    \begin{array}{@{}l@{}c@{}l@{}}
      \intro*\iverson{ & \blank & }_A : \B → A \\
      \iverson{ & \true  & }_A = 1      \\
      \iverson{ & \false & }_A = 0
    \end{array}

    \begin{array}{@{}l@{}}
      \intro*\assoc_{A,B,C} : A × (B × C) → (A × B) × C \\
      \assoc_{A,B,C}(x, (y, z)) = ((x, y), z)
    \end{array}
  \end{mathpar}
  \vspace{-0.2in}
  \Description{Miscellaneous Functions}
  \caption{Miscellaneous Functions}
  \label{fig:miscellaneous-functions}
  \vspace{-0.15in}
\end{figure}

\subsection{Probability Theory} \label{sec:background/probability-theory}

\paragraph{Measure Theory}
\AP
A set $A$ is a ""measurable space"" if it is equipped with
a designated set $\mset_A ⊆ 𝓟{A}$ (called a ""$σ$-algebra"") such that $\intro*\Borel{\mset_A} ⊆ \mset_A$, where
$\Borel{X}$ is the closure of $X$ under complements and countable unions.
We define several (but not all) instances of $\mset$ in \Cref{fig:measurable-sets}.
\begin{figure}[t]
  \begin{align*}
    \mset_{\B}       & = \B                                            &
    \mset_{ℕ}        & = ℕ                                               \\
    \mset_{\ennreal} & = \Borel{\{\, (r, r') | r < r' ∈ \ennreal \,\}} &
    \mset_{\Sample}  & = \Borel{\{\, (r, r') | r < r' ∈ \Sample \,\}}
  \end{align*}
  \vspace{-0.2in}
  \Description{Common $σ$-Algebras.}
  \caption{Common $σ$-Algebras.}
  \label{fig:measurable-sets}
  \vspace{-0.15in}
\end{figure}
The class of measurable spaces is denoted by $\Meas$.
A function $f : A → B$ for $A, B ∈ \Meas$ is ""measurable"" if $\app{f}^{-1}{X} ∈ \mset_A$ for all $X ∈ \mset_B$. The set of "measurable functions" from $A$ to $B$ is denoted by $A \mfun B$.

\paragraph{Quasi Borel Spaces}
\AP
A set $A$ is a ""quasi-Borel space"" if there is a designated set $\Elem{A} ⊆ 𝓟(\ennreal → A)$ of ""random variables"" such that
\begin{enumerate}
  \item $\Elem{A}$ contains all constant functions,
  \item $\Elem{A}$ is closed under precomposition with all "measurable" $f : \ennreal → \ennreal$, and
  \item $\Elem{A}$ is closed under countable case analysis, i.e., $(r ↦ \app{f}(\app{i}{r}){r}) ∈ \Elem{A}$ for all $f : ℕ → \Elem{A}$ and $i : \ennreal \mfun ℕ$.
\end{enumerate}
We define several (but not all) instances of $\Elem$ in \Cref{fig:random-variables}.
\begin{figure}[t]
  \begin{align*}
    \Elem{\B}           & = (\ennreal \mfun \B)                        &
    \Elem{ℕ}            & = (\ennreal \mfun ℕ)                           \\
    \Elem{\ennreal}     & = (\ennreal \mfun \ennreal)                  &
    \Elem{\Sample}      & = (\ennreal \mfun \Sample)                     \\
    \Elem{∏_{i ∈ I} Aᵢ} & = ∏_{i ∈ I} \Elem{Aᵢ}                        &
    \Elem(A \qfun B)    & = \{ \curry{f} | f : \ennreal × A \qfun B \}
  \end{align*}
  \begin{equation*}
    \Elem{∑_{i ∈ I} Aᵢ} = \left\{\, (r ↦ \inj_{(\app{e}(\app{i}{r}))}(\app{α}_{(\app{i}{r})}{r})) | i : \ennreal \qfun ℕ, e : ℕ → I, αᵢ ∈ \Elem{A_{(\app{e}{i})}} \,\right\}
  \end{equation*}
  \vspace{-0.2in}
  \Description{Common Sets of Random Variables.}
  \caption{Common Sets of Random Variables.}
  \label{fig:random-variables}
  \vspace{-0.15in}
\end{figure}
The class of quasi-Borel spaces is denoted by $\Qbs$.
A function $f : A → B$ on $A, B ∈ \Qbs$ is ""quasi-measurable"" if $α ∈ \Elem{A}$ implies $(f ∘ α) ∈ \Elem{B}$. The set of "quasi-measurable" functions from $A$ to $B$ is denoted by $A \qfun B$. A function $f : A × \ennreal \qfun B$ is called a ""random function"".

\paragraph{Functors, Monads, and Arrows.}
\AP
A ""functor"" is a mapping $F : \Qbs → \Qbs$ equipped with a designated function $\map_F : (A \qfun B) \qfun (\app{F}{A} \qfun \app{F}{B})$ such that
\begin{enumerate}
  \item $\map_F{\id_A} = \id_{\app{F}{A}}$ and
  \item $\map_F(f ∘ g) = \map_F{f} ∘ \map_F{g}$ for all $f : B \qfun C$ and $g : A \qfun B$.
\end{enumerate}
\AP
A (strong) ""pre-monad"" is a mapping $F : \Qbs → \Qbs$ equipped with designated operations
\phantomintro\unit
\phantomintro\bind
\begin{align*}
  \reintro*\unit_F & : A \qfun \app{F}{A}                                        &
  \reintro*\bind_F & : (A \qfun \app{F}{B}) \qfun (\app{F}{A} \qfun \app{F}{B}).
\end{align*}
A ""monad"" is a "pre-monad" $F$ such that, for all $f : B \qfun C$ and $g : A \qfun B$,
\begin{align*}
  \bind_F{\unit_F}     & = \id &
  \bind_F{f} ∘ \unit_F & = f   &
  \bind_F{f} ∘ \bind_F{g} = \bind_F(\bind_F{f} ∘ g)
\end{align*}
Every "monad" $F$ is also a "functor" by defining $\map_F{f} = \bind_F(\unit_F ∘ f)$ for functions $f : A \qfun B$.
\AP
A ""pre-arrow"" is a mapping $F : \Qbs × \Qbs → \Qbs$ equipped with designated operations
\phantomintro\arr
\phantomintro\afirst
\phantomintro\acomp
\phantomintro\aleft
\begin{align*}
  \reintro*\arr_F     & : (A \qfun B) \qfun \app{F}(A, B)                      &
  \reintro*\afirst_F  & : \app{F}(A, B) \qfun \app{F}(A × C, B × C)              \\
  (\reintro*\acomp_F) & : \app{F}(B, C) × \app{F}(A, B) \qfun \app{F}(A, C)    &
  \reintro*\aleft_F   & : \app{F}(A, B) \qfun \app{F}(A \bisum C, B \bisum C).
\end{align*}
An ""arrow"" is a "pre-arrow" $F$ which satisfies the properties in \Cref{fig:arrows}.
We omit the subscripts on $\map$, $\unit$, $\bind$, $\arr$, $\acomp$, $\afirst$, and $\aleft$ when they can be inferred from the context.

\begin{figure}
  \begin{align*}
    f \acomp \arr{\id}                      & = f                                    &
    (f \acomp g) \acomp h                   & = f \acomp (g \acomp h)                  \\
    \arr(f ∘ g)                             & = \arr{f} \acomp \arr{g}               &
    \afirst(\arr{f})                        & = \arr(f \fprod \id)                     \\
    \afirst(f \acomp g)                     & = \afirst{f} \acomp \afirst{g}         &
    \arr(\id \fprod g) \acomp \afirst{f}    & = \afirst{f} \acomp \arr(\id \fprod g)   \\
    \arr(\blank₁) \acomp \afirst{f}         & = f \acomp \arr(\blank₁)               &
    \arr{\assoc} \acomp \afirst(\afirst{f}) & = \afirst{f} \acomp \arr{\assoc}
  \end{align*}
  \vspace{-0.2in}
  \Description{Arrow Properties}
  \caption{Arrow Properties.}
  \label{fig:arrows}
  \vspace{-0.15in}
\end{figure}

\paragraph{Measures and Distributions.}
\AP
\citet{quasi-borel-spaces} construct a ""distribution monad"" (or \reintro{measure monad}) $\Dist : \Qbs → \Qbs$ equipped with an ""integration"" operator $(\integral{\blank}{\blank}) : (A \mfun \ennreal) × \Dist{A} \mfun \ennreal$ for all $A ∈ \Qbs$. We write $\integral[x]{y}{μ}$ in place of $\integral{(x ↦ y)}{μ}$ when $μ ∈ \Dist{A}$.
\ifextended
  The "integration operator" satisfies the properties given in \Cref{fig:integral} for all $f, f₁, f₂ : A \qfun \ennreal$, $x ∈ A$, $c ∈ \ennreal$, $μ ∈ \Dist{A}$, $ν ∈ \Dist{B}$, $g : B \qfun A$, and $k₁ ≤ k₂ ≤ … : A \qfun \ennreal$ where $kᵢ$ are ordered point-wise and converge to a function $k : A \qfun \ennreal$.
\else
  The "integration operator" satisfies standard properties.
\fi
Given $f : A \qfun B$, $c ∈ \ennreal$, and $μ ∈ \Dist{A}$, the distributions $\map{f}{μ} ∈ \Dist{B}$ (called the ""push-forward"" of $f$ on $μ$) and $c \intro*\dmul μ ∈ \Dist{A}$ satisfy the properties
\begin{align*}
  \integral{g}{(\map{f}{μ})} & = \integral{g ∘ f}{μ}     &
  \integral{h}{(c \dmul μ)}  & = c \dmul \integral{h}{μ}
\end{align*}
for all $g : B \qfun \ennreal$ and $h : A \qfun \ennreal$.

\ifextended
  \begin{figure}[t]
    \begin{align*}
      \integral{f}{(\unit{x})}                   & = \app{f}{x}                                  &
      \integral{f}{(\bind{g}{ν})}                & = \integral[x]{\integral{f}{(\app{g}{x})}}{ν}   \\
      \integral[x]{c ⋅ \app{f}{x}}{μ}            & = c ⋅ \integral{f}{μ}                         &
      \integral[x]{\app{f₁}{x} + \app{f₂}{x}}{μ} & = \integral{f₁}{μ} + \integral{f₂}{μ}           \\
      \lim_{n \to ∞} \integral{kₙ}{μ}            & = \integral{k}{μ}
    \end{align*}
    \Description{Integral Properties.}
    \caption{Integral Properties.}
    \label{fig:integral}
  \end{figure}
\fi

\AP
The ""uniform distribution"" on samples $\uniform ∈ \Dist{\Sample}$ is the (unique) "distribution" satisfying $\integral[σ]{\iverson{σ₁ ≤ σ ≤ σ₂}}{\uniform} = σ₂ - σ₁$ for all $σ₁ ≤ σ₂ ∈ \Sample$. There exists two functions $\intro*\projl, \intro*\projr : \Sample → \Sample$ such that
\begin{align*}
  \integral[σ]{\app{f}(\projl{σ}, \projr{σ})}{\uniform} = \integral[σ₁]{\integral[σ₂]{\app{f}(σ₁, σ₂)}{\uniform}}{\uniform}
\end{align*}
for all $f : \Sample × \Sample \qfun \ennreal$.

\paragraph{Lists and Trees.}
\AP
The quasi-Borel spaces of ""lists"" $\List{A}$ and ""(rose) trees"" $\Tree{A}$ over $A ∈ \Qbs$ are the smallest sets satisfying
\phantomintro{ε}
\phantomintro{∷}
\phantomintro\node
\begin{align*}
  \List{A} & = \{ \reintro*{ε} \} ∪ \{ x \reintro*∷ \pl{x} | x ∈ A, \pl{x} ∈ \List{A} \} &
  \Tree{A} & = \{ \reintro*\node(x, \pl{x}) | x ∈ A, \pl{x} ∈ \List(\Tree{A}) \}.
\end{align*}
The functions $((x, \pl{x}) ↦ x ∷ \pl{x})$, $((x, \pl{x}) ↦ \node(x, \pl{x}))$, and any functions built using structural recursion with "quasi-measurable" functions, are themselves "quasi-measurable". The mappings $\List, \Tree : \Qbs → \Qbs$ form "monads" with the operations defined in \Cref{fig:lists-and-trees}.

\begin{figure}[t]
  \textbf{Monad Operations}
  \begin{align*}
     & \begin{array}[t]{@{}l@{}}
         \unit : A \qfun \List{A} \\
         \unit{a} = a ∷ ε
       \end{array}                                                                                     &
     & \begin{array}[t]{@{}l@{}}
         \bind : (A \qfun \List{B}) \qfun \List{A} \qfun \List{B} \\
         \bind{f}{ε} = ε                                          \\
         \bind{f}(a ∷ \pl{a}) = \app{f}{a} ⧺ \bind{f}{\pl{a}}
       \end{array}                                         \\
     & \begin{array}[t]{@{}l@{}}
         \unit : A \qfun \Tree{A} \\
         \unit{a} = \node(a, ε)
       \end{array}                                                                                     &
     & \begin{array}[t]{@{}l@{}}
         \bind : (A \qfun \Tree{B}) \qfun \Tree{A} \qfun \Tree{B}               \\
         \bind{f}(\node(a, \pl{a})) = \node(b, \pl{b} ⧺ \map(\bind{f}){\pl{a}}) \\
         \quad \text{where } \node(b, \pl{b}) = \app{f}{a}
       \end{array}
  \end{align*}
  \textbf{Auxiliary Definitions}
  \[
    \begin{array}{@{}l@{}}
      (\intro*⧺) : \List{A} × \List{A} \qfun \List{A} \\
      ε ⧺ \pl{b} = \pl{b}                             \\
      (a ∷ \pl{a}) ⧺ \pl{b} = a ∷ (\pl{a} ⧺ \pl{b})
    \end{array}
  \]
  \vspace{-0.2in}
  \Description{Lists and Rose Trees.}
  \caption{Lists and Rose Trees.}
  \label{fig:lists-and-trees}
  \vspace{-0.15in}
\end{figure}

\end{scope}

%% file: section/3-analysis.tex
\begin{scope}
\knowledgeimport{hedgehog}
\knowledgeimport{background}

\section{Hedgehog} \label{sec:analysis}

\begin{outline}[b]
  \textbf{Hedgehog}
  \begin{itemize}
    \item Syntax
          \begin{itemize}
            \item Types and Terms
            \item Coin Examples
            \item Typing Relation
          \end{itemize}
    \item Semantics
          \begin{itemize}
            \item Interpretations and Semantics
            \item Sampling Interpretation
            \item Coin Examples (coarse-grained)
            \item Semantic Equivalence
            \item Common Properties
          \end{itemize}
    \item Distributions
          \begin{itemize}
            \item Distribution Semantics
            \item Coin Examples (fine-grained)
            \item Non-compositionality
          \end{itemize}
    \item Contextual Equivalence
          \begin{itemize}
            \item Contexts and Contextual Equivalence
            \item Soundness and Completeness
            \item Sampling Interpretation is Sound and Complete
          \end{itemize}
  \end{itemize}
\end{outline}
\AP
In this section, we provide a formal account of the \hedgehog{} \citep{hedgehog} language. Our formalization is given in the style of the monadic metalanguage by \citet{monads}. At a high level, \hedgehog{} is a probabilistic programming language, where generators are specified using probabilistic operators. In addition to probabilistic behaviour, \hedgehog{} contains operators for specifying shrinking behaviour.

\AP
Generator optimizations require formal justification relative to an appropriate semantics. We formalize two candidate semantics for \hedgehog{}:
\begin{enumerate*}
  \item a "sampling semantics" (\Cref{def:sampling-interpretation}), where generators represent "random variables", which closely mirrors \hedgehog{}'s actual implementation, and
  \item a "distribution semantics" (\Cref{def:distribution}), where generators represent distributions, which more closely model how users actually think about generators.
\end{enumerate*}
Both of these semantics have significant issues which make them unsuitable for equivalence proofs. The "sampling semantics" is overly \emph{fine-grained}, i.e., it distinguishes many "terms" which users expect to be identical (\Cref{ex:coin-semantics}) and invalidates many common program transformations (\Cref{thm:unsound-transformations}). The "distribution semantics" is not "compositional", i.e., the distribution of a term is not uniquely determined by the distributions of its immediate sub-"terms" (\Cref{prop:non-compositionality}), which complicates reasoning about larger programs. Unfortunately, this state of affairs cannot be fixed by simply finding the ``right semantics'': we derive "soundness" and "completeness" conditions from the "distribution semantics" (\Cref{def:soundness-and-completeness}) and show that \emph{every} ("sound" and "complete") "compositional" semantics is equivalent to the "sampling semantics" (\Cref{corr:sampling-correct}). This result implies that several common program transformation are unsound, which justifies our need for a restricted version of \hedgehog{}.

\subsection{Syntax}

\AP
Syntactic elements are written in \texttt{teletype} font.

\begin{definition}[\hedgehog{} Syntax] \label[definition]{def:syntax}
  \AP
  \phantomintro\tabsurd
  \phantomintro\tunit
  \phantomintro\tinl
  \phantomintro\tinr
  \phantomintro\tmatch
  \phantomintro\tpair
  \phantomintro\tfst
  \phantomintro\tsnd
  \phantomintro\tlam
  \phantomintro{function application}
  \phantomintro\treturn
  \phantomintro\tlet
  \phantomintro\tflip
  \phantomintro\tshrink
  \looseness-1 The set of \hedgehog{} ""types"" $\Type$ is generated by
  \begin{equation*}
    τ
    ⩴ \intro*\Empty
    | \intro*\Unit
    | \intro*\Sum{τ}{τ}
    | \intro*\Prod{τ}{τ}
    | \intro*\Fun{τ}{τ}
    | \intro*\Gen{τ}
  \end{equation*}
  and the set of \hedgehog{} ""terms"" $\Term$ is generated by
  \begin{align*}
    t
    ⩴ {} & x
    | \reintro*\tabsurd{t}
    | \reintro*\tunit
    | \reintro*\tinl{t}
    | \reintro*\tinr{t}
    | \reintro*\tmatch{t}{x}{t}{x}{t}
    | \reintro*\tpair{t, t}  \\
    | {} & \reintro*\tfst{t}
    | \reintro*\tsnd{t}
    | \reintro*\tlam{x}{τ}{t}
    | \app{t}{t}
    | \reintro*\treturn{t}
    | \reintro*\tlet{x}{t}{t}
    | \reintro*\tflip{p}
    | \reintro*\tshrink
  \end{align*}
  where ""variables"" (e.g., $x$) are sourced from a countably infinite set of names $\Var$, and probabilities (e.g., $p$) are elements of the real interval $[0, 1]$. We identify terms up to $α$-equivalence.
\end{definition}

\AP The set of types consists of the "empty type" ($\Empty$), the "singleton type" ($\Unit$), "sum types" ($\Sum{τ}{τ}$), "product types" ($\Prod{τ}{τ}$), "function types" ($\Fun{τ}{τ}$), and "generator types" ($\Gen{τ}$). Intuitively, the type $\Gen{τ}$ is the type of generators which produce elements of the type $τ$. The set of terms contains many standard constructs including variables ($x$), left and right injections ($\tinl$, $\tinr$), eliminations for values of the empty type ($\tabsurd{t}$), case analysis ($\tmatch{t}{x}{t}{x}{t}$), the singleton value ($\tunit$), tuples ($\tpair{t, t}$), anonymous functions ($\tlam{x}{τ}{t}$), function application ($\app{t}{t}$), effect-free computations ($\treturn{t}$), and sequential composition ($\tlet{x}{t}{t}$). We derive some common constructs in \Cref{fig:derived-constructs}. In addition to standard constructs, \hedgehog{} has the following effectful operations:
\begin{enumerate}
  \item The biased \reintro{coin flip operation} $\tflip{p}$ generates $\ttrue$ with probability $p$ and $\tfalse$ with probability $1 - p$.
  \item The \reintro{shrinking operation} $\tshrink$ may be applied to an $n$-tuple $\tpair{t₀, t₁, …, tₙ}$. The operation $\tshrink{\tpair{t₀, t₁, … tₙ}}$ behaves like $t₀$ during the "testing stage@@introduction", but attempts each of $t₁, …, tₙ$ in order during the "shrinking stage@@introduction".
\end{enumerate}

\begin{figure}[t]
  \begin{mathpar}
    \intro*\Bool = \Sum{\Unit}{\Unit}

    \intro*\ttrue = \tinl{\tunit}

    \intro*\tfalse = \tinr{\tunit}

    \intro*\tif{t₁}{t₂}{t₃} = \tmatch{t₁}{x}{t₂}{x}{t₃}

    \reintro*\tpair{t₁, t₂, …, tₙ} = \tpair{t₁, \tpair{t₂, …, tₙ}}

    \intro*\tnot{t₁} = \tif{t₁}{\tfalse}{\ttrue}

    t₁ \intro*\teq t₂ = \tif{t₁}{t₂}{\tnot{t₂}}

    t₁ \intro*\tand t₂ = \tif{t₁}{t₂}{\tfalse}
  \end{mathpar}
  \Description{Derived}
  \vspace{-0.2in}
  \caption{Derived Constructs.}
  \label{fig:derived-constructs}
  \vspace{-0.15in}
\end{figure}

\begin{example} \label[example]{ex:coin-syntax}
  \AP
  \phantomintro\tcoini
  \phantomintro\tcoinii
  We recast our examples from \Cref{sec:introduction} in our chosen syntax, modelling \texttt{bool\textunderscore} as $\tflip{\nicefrac{1}{2}}$:
  \begin{align*}
    \begin{array}[t]{@{}l@{}}
      \reintro*\tcoini =
      \tlet+{x}{\tflip{\nicefrac{1}{2}}}{
        \tlet+{y}{\tflip{\nicefrac{1}{2}}}{
          \quad \treturn{x \teq y}}}
    \end{array}
     &  &
    \begin{array}[t]{@{}l@{}}
      \reintro*\tcoinii =
      \tlet+{x}{\tflip{\nicefrac{1}{2}}}{
        \quad \treturn{x}}
    \end{array}
  \end{align*}
\end{example}

\begin{definition}[\hedgehog{} Environments and Typing] \label[definition]{def:typing}
  \AP
  An ""environment"" $Γ ∈ ∏_{x ∈ \Dom{Γ}} \Type$ is a family mapping "variables" from a finite ""domain"" $\Dom{Γ} ⊆ \Var$ to "types". The set of environments is denoted by $\Env$, and the singleton environment is written $x : τ$. The ""typing relation"" $(\typed[\blank]{\blank}{\blank}) ⊆ \Env × \Term × \Type$ is the smallest relation satisfying the rules in \Cref{fig:typing}. If $t$ has no free variables, we write $\typed{t}{τ}$ in place of $∀Γ, \typed[Γ]{t}{τ}$, since the environment has no bearing on whether $t$ is typed.
\end{definition}

\begin{figure}[t]
  \begin{mathpar}
    \infer{ }{\typed[Γ, x : τ]{x}{τ}}

    \infer{
      \typed[Γ]{t}{τ₁}
    }{\typed[Γ]{\tinl{t}}{\Sum{τ₁}{τ₂}}}

    \infer{
      \typed[Γ]{t}{τ₂}
    }{\typed[Γ]{\tinr{t}}{\Sum{τ₁}{τ₂}}}

    \infer{
      \typed[Γ]{t}{\Empty}
    }{\typed[Γ]{\tabsurd{t}}{τ}}

    \infer{
      x₁, x₂ ∉ \Dom{Γ} \\
      \typed[Γ]{t₀}{\Sum{τ₁}{τ₂}} \\\\
      \typed[Γ, x₁ : τ₁]{t₁}{τ₃} \\
      \typed[Γ, x₂ : τ₂]{t₂}{τ₃}
    }{\typed[Γ]{\tmatch{t₀}{x₁}{t₁}{x₂}{t₂}}{τ_3}}

    \infer{ }{\typed[Γ]{\tunit}{\Unit}}

    \infer{
      \typed[Γ]{t₁}{τ₁} \\
      \typed[Γ]{t₂}{τ₂}
    }{\typed[Γ]{\tpair{t₁, t₂}}{\Prod{τ₁}{τ₂}}}

    \infer{
      \typed[Γ]{t}{\Prod{τ₁}{τ₂}}
    }{\typed[Γ]{\tfst{t}}{τ₁}}

    \infer{
      \typed[Γ]{t}{\Prod{τ₁}{τ₂}}
    }{\typed[Γ]{\tsnd{t}}{τ₂}}

    \infer{
      x ∉ \Dom{Γ} \\
      \typed[Γ, x : τ₁]{t}{τ₂}
    }{\typed[Γ]{\tlam{x}{τ₁}{t}}{\Fun{τ₁}{τ₂}}}

    \infer{
      \typed[Γ]{t}{τ}
    }{\typed[Γ]{\treturn{t}}{\Gen{τ}}}

    \infer{
    }{\typed[Γ]{\tflip{p}}{\Gen{\Bool}}}

    \infer{
      x ∉ \Dom{Γ} \\
      \typed[Γ]{t₁}{\Gen{τ₁}} \\
      \typed[Γ, x : τ₁]{t₂}{\Gen{τ₂}}
    }{\typed[Γ]{\tlet{x}{t₁}{t₂}}{\Gen{τ₂}}}

    \infer{
    }{\typed[Γ]{\tshrink}{\Fun{\Prod{τ}{\Prod{⋯}{τ}}}{\Gen{τ}}}}
  \end{mathpar}
  \vspace{-0.2in}
  \caption{Hedgehog's Typing Rules.}
  \Description{\hedgehog{}'s Typing Rules.}
  \label{fig:typing}
  \vspace{-0.15in}
\end{figure}

\AP
Not every term is sensible (e.g., consider `$\tfst{\tunit}$'). The typing relation $\typed[Γ]{t}{τ}$ describes exactly which terms are sensible within a given environment.

\subsection{Semantics}

\begin{definition}[\hedgehog{} Interpretation and Semantics] \label[definition]{def:semantics}
  \AP
  \phantomintro\flip
  \phantomintro\shrink
  A \hedgehog{} ""interpretation"" is a "pre-monad" $F : \Qbs → \Qbs$ equipped with additional operations
  \begin{align*}
     & \reintro*\flip_F : [0, 1] \qfun \app{F}{\B}            &
     & \reintro*\shrink^n_F : ∏_{i = 0}^n A \qfun \app{F}{A}.
  \end{align*}
  An interpretation induces
  \begin{enumerate*}
    \item a ""type and environment semantics"" $\tysem{\blank}^F : \Type ∪ \Env → \Qbs$ and
    \item for all $\typed[Γ]{t}{τ}$ a ""term semantics"" $\tmsem[Γ]{t}[τ]^F : \tysem{Γ}^F \qfun \tysem{τ}^F$
  \end{enumerate*}
  (see \Cref{fig:semantics}). When $t$ has no free variables, we write $\tmsem{t}[τ]^F$ in place of $\tmsem[\sing]{t}[τ]^F\sing$. We also omit the type ascription ($: τ$) when it can be inferred from the surrounding context.
\end{definition}

\begin{figure}[t]
  \textbf{Type and Environment Semantics}
  \begin{mathpar}
    \tysem{\Empty}^F = \zero

    \tysem{\Unit}^F = \one

    \tysem{\Sum{τ₁}{τ₂}}^F = \tysem{τ₁}^F + \tysem{τ₂}^F

    \tysem{\Prod{τ₁}{τ₂}}^F = \tysem{τ₁}^F × \tysem{τ₂}^F

    \tysem{\Fun{τ₁}{τ₂}}^F = \tysem{τ₁}^F \qfun \tysem{τ₂}^F

    \tysem{\Gen{τ}}^F = \app{F}{\tysem{τ}^F}

    \tysem{Γ}^F = ∏_{x ∈ \Dom{Γ}} \tysem{Γ_x}^F
  \end{mathpar}
  \textbf{Term Semantics}
  \begin{mathpar}
    \tmsem[Γ, x : τ]{x}[τ]^F{ρ} = ρ_x

    \tmsem[Γ]{\tunit}[\Unit]^F{ρ} = \sing

    \tmsem[Γ]{\tabsurd{t}}[τ]^F{ρ} = \absurd(\tmsem[Γ]{t}[\Empty]^F{ρ})

    (\tmsem[Γ]{\tfst{t}}[τ₁]^F{ρ}, \tmsem[Γ]{\tsnd{t}}[τ₂]^F{ρ}) = \tmsem[Γ]{t}[\Prod{τ₁}{τ₂}]^F{ρ}

    \tmsem[Γ]{\tpair{t₁, t₂}}[\Prod{τ₁}{τ₂}]^F{ρ} = (\tmsem[Γ]{t₁}[τ₁]^F{ρ}, \tmsem[Γ]{t₂}[τ₂]^F{ρ})

    \tmsem[Γ]{\tinl{t}}[\Sum{τ₁}{τ₂}]^F = \inj_1 ∘ \tmsem[Γ]{t}[τ₁]^F

    \tmsem[Γ]{\tinr{t}}[\Sum{τ₁}{τ₂}]^F = \inj_2 ∘ \tmsem[Γ]{t}[τ₂]^F

    \begin{array}{@{}l@{}}
      \tmsem[Γ]{\tmatch{t₀}{x₁}{t₁}{x₂}{t₂}}[τ₃]^F{ρ}                                                                                            \\
      \qquad = \begin{cases}
                 \tmsem[Γ, x₁ : τ₁]{t₁}[τ₃]^F(ρ_x)_{x ∈ \Dom{Γ} ∪ \{ x₁ \}} & \text{if } \tmsem[Γ]{t₀}[\Sum{τ₁}{τ₂}]^F{ρ} = \app{\inj_1}{ρ_{x₁}} \\
                 \tmsem[Γ, x₂ : τ₂]{t₂}[τ₃]^F(ρ_x)_{x ∈ \Dom{Γ} ∪ \{ x₂ \}} & \text{if } \tmsem[Γ]{t₀}[\Sum{τ₁}{τ₂}]^F{ρ} = \app{\inj_2}{ρ_{x₂}}
               \end{cases} \\
      \quad \text{where } \typed[Γ]{t₀}{\Sum{τ₁}{τ₂}},\quad \typed[Γ, x₁ : τ₁]{t₁}{τ₃},\quad \typed[Γ, x₂ : τ₂]{t₂}{τ₃}
    \end{array}

    \tmsem[Γ]{\tlam{x}{τ₁}{t}}[\Fun{τ₁}{τ₂}]^F{ρ} = ρ_x ↦ \tmsem[Γ, x : τ₁]{t}[τ₂]^F(ρ_y)_{y ∈ \Dom{Γ} ∪ \{x\}}

    \begin{array}{@{}l@{}}
      \tmsem[Γ]{\app{t₁}{t₂}}[τ₂]^F{ρ} = \app{\tmsem[Γ]{t₁}[\Fun{τ₁}{τ₂}]^F{ρ}}(\tmsem[Γ]{t₂}[τ₁]^F{ρ}) \\
      \quad \text{where } \typed[Γ]{t₁}{\Fun{τ₁}{τ₂}},\quad \typed[Γ]{t₂}{τ₁}
    \end{array}

    \tmsem[Γ]{\treturn{t}}[\Gen{τ}]^F{ρ} = \unit(\tmsem[Γ]{t}[τ]^F{ρ})

    \begin{array}{@{}l@{}}
      \tmsem[Γ]{\tlet{x}{t₁}{t₂}}[\Gen{τ₂}]^F{ρ}                                                                                 \\
      \qquad = \bind_F(ρ_x ↦ \tmsem[Γ, x : τ₁]{t₂}[\Gen{τ₂}]^F\app(ρ_y)_{y ∈ \Dom{Γ} ∪ \{x\}})(\tmsem[Γ]{t₁}[\Gen{τ₁}]^F\app{ρ}) \\
      \quad \text{where } \typed[Γ]{t₁}{\Gen{τ₁}}, \quad \typed[Γ, x : τ₁]{t₂}{\Gen{τ₂}}
    \end{array}

    \tmsem[Γ]{\tflip{p}}[\Gen{\Bool}]^F{ρ} = \flip_F{p}

    \tmsem[Γ]{\tshrink}[\Fun{\underbrace{\Prod{τ}{\Prod{⋯}{τ}}}_{n + 1\text{ times}}}{\Gen{τ}}]^F{ρ} = \shrink^n_F
  \end{mathpar}
  \vspace{-0.15in}
  \caption{Semantics of \hedgehog{}.}
  \Description{Semantics of \hedgehog{}.}
  \vspace{-0.1in}
  \label{fig:semantics}
\end{figure}

\AP
The semantics of most "types" and "terms" are standard, where "types" denote suitable "quasi-Borel spaces" and "terms" denote "quasi-measurable functions" over \reintro{variable assignments} (i.e., elements of $\tysem{Γ} = ∏_{x ∈ \Dom{Γ}} Γₓ$). An interpretation specifies the semantics of generator "types" and "terms" ($\Gen$, $\treturn$, $\tletkeyword$, $\tflip$, and $\tshrink$). We only require that an interpretation $F$ is a "pre-monad", and not a "monad": QuickCheck generators are known not to form a monad \citep{quickcheck}, and \hedgehog{}'s "sampling interpretation" (\Cref{def:sampling-interpretation}) inherits this property. An important feature of \Cref{def:semantics} is that it is \reintro{compositional}, i.e., the semantics of a "term" is uniquely determined by the semantics of its immediate sub-"terms".

\begin{definition}[\hedgehog{} Sampling Interpretation] \label[definition]{def:sampling-interpretation}
  \AP
  The ""sampling interpretation"" $\Sampling : \Qbs → \Qbs$ of Hedgehog is defined in \Cref{fig:sampling-interpretation}.
\end{definition}

\begin{figure}[t]
  \begin{mathpar}
    \Sampling{A} = (\Sample \qfun \Tree{A})

    \unit_{\Sampling}{a} = σ ↦ \unit{a}

    \bind_{\Sampling}{f}{m} = σ ↦ \bind(a ↦ \app{f}{a}(\projr{σ}))(\app{m}(\projl{σ}))

    \flip_{\Sampling}{p} = σ ↦ \unit(σ < p)

    \shrink^n_{\Sampling}(a₀, a₁, …, aₙ) = σ ↦ \node(a₀, \unit{a₁} ∷ … ∷ \unit{aₙ} ∷ ε)
  \end{mathpar}
  \vspace{-0.2in}
  \caption{\hedgehog{}'s Sampling Interpretation.}
  \Description{Hedgehog's Sampling Interpretation.}
  \label{fig:sampling-interpretation}
  \vspace{-0.15in}
\end{figure}

\AP
\Cref{def:sampling-interpretation} closely mirrors the actual implementation of \hedgehog{}. We describe each component of the "sampling interpretation" as follows:
\begin{enumerate}
  \item Generators ($\Gen{τ}$) are represented as deterministic functions from the "sample space" $\Sample$ to "rose trees" (i.e., generators are "random variables"). Intuitively, the probabilistic behaviour of a generator arises by feeding it a uniformly distributed random "sample", and the resulting "tree" (called a ""shrink tree"") describes a value and all the ways it can be shrunk.
  \item The ``pure'' computation $\treturn{t}$ ignores its "sample" and returns a tree with a single node labelled by the value of $t$.
  \item Sequential composition $\tlet{x}{t₁}{t₂}$ splits its "sample", passing a different half to $t₁$ and $t₂$.
  \item The probabilistic choice operation $\tflip{p}$ returns $\ttrue$ if its "sample" lies in the range $[0, p)$.
  \item The shrink operation $\tshrink$ ignores its "sample" and returns a shrink tree comprising of its arguments.
\end{enumerate}

\begin{example} \label[example]{ex:coin-semantics}
  \AP
  The semantics of $\typed{\tcoini, \tcoinii}{\Gen{\Bool}}$, given in \Cref{ex:coin-syntax}, are as follows:
  \begin{align*}
    \tmsem{\tcoini}^{\Sampling}{σ}  & = \unit(\projl{σ} < \nicefrac{1}{2} ⇔ \projl(\projr{σ}) < \nicefrac{1}{2}) &
    \tmsem{\tcoinii}^{\Sampling}{σ} & = \unit(\projl{σ} < \nicefrac{1}{2})
  \end{align*}
  While $\tcoini$ and $\tcoinii$ intuitively represent the same distribution of values, they are not equivalent under the "sampling semantics": $\Sampling$ is overly \emph{fine-grained}.
\end{example}

\subsection{Program Transformations}

\AP
We often talk about "term" equivalence under an "interpretation" (e.g., does $\tmsem{\tcoini}^F = \tmsem{\tcoinii}^F$ for some other "interpretation" $F$?), so we define a shorthand in \Cref{def:semantic-equivalence}.

\begin{definition}[\hedgehog{} Semantic Equivalence] \label[definition]{def:semantic-equivalence}
  \AP
  Two \hedgehog{} "terms" $\typed[Γ]{t₁, t₂}{τ}$ are ""semantically equivalent"" under an "interpretation" $F$ (written $\semeq{F}[Γ]{t₁}{t₂}[τ]$) if $\tmsem[Γ]{t₁}[τ]^F = \tmsem[Γ]{t₂}[τ]^F$. We write $\semeq{F}{t₁}{t₂}[τ]$ for $\semeq{F}[\sing]{t₁}{t₂}[τ]$ and omit the type ascription ($: τ$) when it can be inferred from the surrounding context.
\end{definition}

To further illustrate how the "sampling interpretation" is unsuitable for semantic equivalence proofs, we show that several common program transformations are not permitted under $\Sampling$.

\begin{theorem} \label[theorem]{thm:unsound-transformations}
  The following equivalences \emph{do not hold} (whenever both sides of the equivalence have the same type):
  \begin{align}
    \semeq{\Sampling}[Γ]{\tlet{x}{\treturn{t}}{\app{f}{x}}
      & }{\app{f}{t}}
    \label{eq:constant-folding} \\
    \semeq{\Sampling}[Γ]{\tlet{x}{t}{\treturn{x}}
      & }{t}
    \label{eq:eta-reduction}    \\
    \semeq{\Sampling}[Γ]{\tlet{y}{(\tlet{x}{t}{\app{f}{x}})}{\app{g}{y}}
      & }{\tlet{x}{t}{\tlet{y}{\app{f}{x}}{\app{g}{y}}}}
    \label{eq:inlining}         \\
    \semeq{\Sampling}[Γ]{\tlet{x}{\tflip{p}}{t}
      & }{t}
    \label{eq:flip-id}          \\
    \semeq{\Sampling}[Γ]{\tlet{x}{\tflip{p}}{\tlet{y}{t}{\app{f}{x}{y}}}
      & }{\tlet{y}{t}{\tlet{x}{\tflip{p}}{\app{f}{x}{y}}}}
    \label{eq:flip-comm}        \\
    \semeq{\Sampling}[Γ]{\tlet{x}{\tflip{p}}{\app{f}(\tnot{x})}
      & }{\tlet{x}{\tflip(1 - p)}{\app{f}{x}}}
    \label{eq:flip-inv}         \\
    \semeq{\Sampling}[Γ]{\tlet{x}{\tflip{p}}{\tlet{y}{\tflip{q}}{\app{f}(x \tand y)}}
      & }{\tlet{x}{\tflip(p ⋅ q)}{\app{f}{x}}}
    \label{eq:flip-conj}
  \end{align}
\end{theorem}

Intuitively, none of these transformations are sound under the "sampling interpretation" because, similar to \Cref{ex:coin-semantics}, they all change how samples are interpreted. We discuss each equation in detail:
\begin{itemize}
  \item \Cref{eq:constant-folding} states that an effect-free computation, when $\tletkeyword$-bound to a variable, can be inlined. This is an optimization: removing a $\tletkeyword$ expression reduces the number of "sample" splits under the "sampling semantics". \Cref{eq:constant-folding} is a form of constant folding \citep{dragon-book}.
  \item \Cref{eq:eta-reduction} formalizes the intuition that $\treturn{x}$ is a ``no-op''. Like \cref{eq:constant-folding}, \cref{eq:eta-reduction} is an optimization and reduces the number of "sample" splits under the "sampling semantics".
  \item \Cref{eq:inlining} states that the order in which multiple sequential compositions are formed does not matter as long as the overall order of computations remains the same. While not strictly an optimization by itself, \cref{eq:inlining} is necessary to justify further optimizations. For example, if \cref{eq:inlining,eq:flip-id} hold then
        \begin{align*}
          \tlet{x}{(\tlet{y}{t₁}{\tflip})}{t₂}
           & = \tlet{y}{t₁}{\tlet{x}{\tflip}{t₂}} \\
           & = \tlet{y}{t₁}{t₂}.
        \end{align*}
        \Cref{eq:constant-folding,eq:eta-reduction,eq:inlining} are essentially the "monad" laws, and hence $\Sampling$ is not a "monad".
  \item \Cref{eq:flip-id} is a standard property of probabilistic choice \citep{probabilistic-λ-calculus-3}, and states that a $\tflip$ statement whose result is unused can be removed. This is also an optimization that reduces the number of "sample" splits under the "sampling semantics".
  \item \Cref{eq:flip-comm} is another standard property of probabilistic choice \citep{commutative-semantics}, and states that the order in which a random value is generated is irrelevant. While not an optimization by itself, \cref{eq:flip-comm} is necessary to justify further optimizations. For example, if \cref{eq:flip-comm,eq:flip-id} hold then
        \begin{align*}
          \left(\tlet+[]{x}{\tflip{p}}{\tlet+{y}{\tshrink{…}}{\tlet+{z}{\tflip{q}}{\quad \app{f}(x \tand z)}}}\right)
          = \left(\tlet+[]{x}{\tflip{p}}{\tlet+{z}{\tflip{q}}{\tlet+{y}{\tshrink{…}}{\quad \app{f}(x \tand z)}}}\right)
          = \left(\tlet+[]{x}{\tflip(p ⋅ q)}{\tlet+{y}{\tshrink{…}}{\quad \app{f}{x}}}\right).
        \end{align*}
  \item \Cref{eq:flip-inv} states that flipping a biased coin with parameter $p$ is the same as flipping a biased coin with parameter $1 - p$ and interchanging the results. Since $p$ is a constant in our setting, \cref{eq:flip-inv} is an optimization because it eliminates a negation ($\tnot$).
  \item \Cref{eq:flip-conj} provides a condition for merging two calls to $\tflip$. Since $\tflip$'s parameter is a constant in our setting, \cref{eq:flip-conj} is an optimization because it reduces the number of "sample" splits under the "sampling semantics".
\end{itemize}

In contrast to \Cref{thm:unsound-transformations}, the following transformations \emph{are} permitted under $\Sampling$.

\begin{theorem} \label[theorem]{thm:sound-transformations}
  The following equivalences hold:
  \begin{align}
    \semeq{\Sampling}[Γ]{\tlet{x}{\treturn{t}}{\tshrink(\app{f₁}{x}){…}({\app{fₙ}{x}})}
      & }{\tshrink(\app{f₁}{t}){…}({\app{fₙ}{t}})}
    \label{eq:shrink-constant-folding} \\
    \semeq{\Sampling}[Γ]{\tlet{x}{\tshrink{t₀}{…}{tₙ}}{\treturn{x}}
      & }{\tshrink{t₀}{…}{tₙ}}
    \label{eq:shrink-eta-reduction}    \\
    \semeq{\Sampling}[Γ]{\tshrink{t}
      & }{\treturn{t}}
    \label{eq:shrink-elim}             \\
    \semeq{\Sampling}[Γ]{\tif{t₁}{t₂}{t₂}
      & }{t₂}
    \label{eq:if-elim}                 \\
    \semeq{\Sampling}[Γ]{\tif{\ttrue}{t₁}{t₂}
      & }{t₁}
    \label{eq:if-true}                 \\
    \semeq{\Sampling}[Γ]{\tif{\tfalse}{t₁}{t₂}
      & }{t₂}
    \label{eq:if-false}                \\
    \semeq{\Sampling}[Γ]{\tflip{0}
      & }{\treturn{\tfalse}}
    \label{eq:flip-zero}               \\
    \semeq{\Sampling}[Γ]{\tflip{1}
      & }{\treturn{\ttrue}}
    \label{eq:flip-one}
  \end{align}
  \begin{equation}
    \semeq{\Sampling}[Γ]{
      \left(
      \tlet*[]{y}{(\tlet*{x}{\tshrink{t₀}{…}{tₙ}}{
      \tshrink(\app{f₁}{x}){…}(\app{fₘ}{x}))}}{
        \tshrink(\app{g₀}{y}){…}(\app{gₖ}{y})}
      \right)
    }{
      \left(
      \tlet+[]{x}{\tshrink{t₀}{…}{tₙ}}{
      \tlet+{y}{\tshrink(\app{f₁}{x}){…}(\app{fₘ}{x})}{
        \quad \tshrink(\app{g₀}{y}){…}(\app{gₖ}{y})}}
      \right)
    }
    \label{eq:shrink-inlining}
  \end{equation}
\end{theorem}
We describe each equation as follows:
\begin{itemize}
  \item \Cref{eq:shrink-constant-folding,eq:shrink-eta-reduction,eq:shrink-inlining} are essentially special cases of \Cref{eq:constant-folding,eq:eta-reduction,eq:inlining} where the computations involved only exhibit shrinking effects.
  \item \Cref{eq:shrink-elim} states that a $\tshrink$ term with only one argument (i.e., no specified shrinks) is identical to a pure computation ($\treturn$).
  \item \Cref{eq:if-elim,eq:if-true,eq:if-false} are common optimizations of $\tifkeyword$ expressions. Unlike the other equations we discuss, these equations hold under any interpretation.
  \item \Cref{eq:flip-zero,eq:flip-one} state that $\tflip$'s behaviour is deterministic when its parameter is either zero or one.
\end{itemize}

From \Cref{thm:unsound-transformations,thm:sound-transformations} we infer that the "sampling interpretation" only permits transformations of \emph{non-probabilistic terms}. Since generators are primarily probabilistic programs, this renders the "sampling interpretation" unsuitable for program optimization.

  Particularly problematic is the failure of \Cref{eq:constant-folding,eq:inlining}, which apply frequently as a result of inlining. For example, consider the program $\tcoiniii$, which is just $\tcoini$ with $\tflip$ replaced with $\tcoinii$:
  \[
    \intro*\tcoiniii = \tlet{x}{\tcoinii}{\tlet{y}{\tcoinii}{\treturn{x \teq y}}}.
  \]
  Aggressively applying \Cref{eq:constant-folding,eq:inlining} gives the following sequence of simplifications:
  \begin{align*}
    \tcoiniii & \mathrel{\kl[\semeq]{=}^{\Sampling}}
    \tlet{x}{\tcoinii}{\tlet{y}{\tcoinii}{\treturn{x \teq y}}} \tag{by definition} \\
                   & \mathrel{\kl[\semeq]{=}^{\Sampling}}
    \tlet+{x}{(\tlet{z}{\tflip{\nicefrac{1}{2}}}{\treturn{z}})}{
      \tlet+{y}{(\tlet{w}{\tflip{\nicefrac{1}{2}}}{\treturn{w}})}{
        \quad \treturn{x \teq y}}} \tag{by definition} \\
                   & \mathrel{\kl[\semeq]{=}^{\Sampling}}
    \tlet+{z}{\tflip{\nicefrac{1}{2}}}{
      \tlet+{x}{\treturn{z}}{
      \tlet+{w}{\tflip{\nicefrac{1}{2}}}{
      \tlet+{y}{\treturn{w}}{
        \quad \treturn{x \teq y}}}}} \tag{by \Cref{eq:inlining}} \\
                   & \mathrel{\kl[\semeq]{=}^{\Sampling}}
    \tlet+{z}{\tflip{\nicefrac{1}{2}}}{
      \tlet+{w}{\tflip{\nicefrac{1}{2}}}{
      \quad \treturn{z \teq w}}} \tag{by \Cref{eq:constant-folding}}
  \end{align*}
  The optimized term is faster because the original term will split the input seed four times, while the optimized term will only split it twice. \citet{goldstein-2026} identify random number generation (which includes seed splitting) as the ``hottest part of the hot path'' in generator execution, so the failure of \Cref{eq:constant-folding,eq:inlining} means we lose out on a potentially high impact optimization.

\subsection{Distributions}

\AP
We have presented in \Cref{def:semantics,def:sampling-interpretation} a "sampling semantics" which closely mirrors the actual implementation of Hedgehog. However, the specifics of how "samples" are turned into values are irrelevant from the user's perspective. Users are more concerned with a generator's "distribution", which determines how well the input space of the SUT is covered. To this end, the "sampling semantics" is not suitable for proving optimization soundness because it is too fine-grained (e.g., see \Cref{ex:coin-semantics}).

\begin{definition}[\hedgehog{} Distribution Semantics] \label[definition]{def:distribution}
  \AP
  The ""distribution semantics"" of any $\typed{t}{\Gen{\Bool}}$ is given by the function
  \[
    \begin{array}{@{}l@{}c@{}l@{}}
      \dist{ & \blank & } : \Term → \Dist{\B} \\
      \dist{ & t      & } = \map{\tmsem{t}^{\Sampling}}{\uniform}.
    \end{array}
  \]
\end{definition}

\AP
Since generators represent "random variables", we obtain a "distribution" via the "push-forward" operation.

\begin{example} \label[example]{ex:coin-distribution}
  \AP
  We compute the "distributions" of $\typed{\tcoini, \tcoinii}{\Gen{\Bool}}$ given in \Cref{ex:coin-syntax}:
  \begin{align*}
    \dist{\tcoini}  & = \nicefrac{1}{2} ⋅ \unit(\unit{\true}) + \nicefrac{1}{2} ⋅ \unit(\unit{\false}) \\
    \dist{\tcoinii} & = \nicefrac{1}{2} ⋅ \unit(\unit{\true}) + \nicefrac{1}{2} ⋅ \unit(\unit{\false})
  \end{align*}
  Terms $\tcoini$ and $\tcoinii$ both produce single-node trees. Each tree is labelled by either $\true$ or $\false$ with probability $\nicefrac{1}{2}$.
\end{example}

\AP
\phantomintro\tcontext
Building on \Cref{ex:coin-distribution}, we show that \hedgehog{}'s "distribution semantics" is non-"compositional", i.e., it cannot be represented as an "interpretation". Define
\[
  \begin{array}{@{}l@{}}
    \typed{\tcontext}{\Fun{\Gen{\Bool}}{\Gen{\Bool}}} \\
    \reintro*\tcontext = \tlam{m}{\Gen{\Bool}}{
                           \tlet{b}{\tshrink{\tpair{\ttrue, \tfalse}}}{
                             \tif{b}{m}{\tcoini}}}
  \end{array}
\]
and consider the "sampling semantics" of $\app{\tcontext}{\tcoini}$ and $\app{\tcontext}{\tcoinii}$. On one hand, we have
\begin{align*}
  \app{\tcontext}{\tcoini}
   & \mathrel{\kl[\semeq]{=}}^{\Sampling} \tlet{b}{\tshrink{\tpair{\ttrue, \tfalse}}}{\tif{b}{\tcoini}{\tcoini}} \\
   & \mathrel{\kl[\semeq]{=}}^{\Sampling} \tlet{b}{\tshrink{\tpair{\ttrue, \tfalse}}}{\tcoini}.
\end{align*}
During shrinking, the "sample" fed to $\tcoini$ does not change, so any tree produced by $\app{\tcontext}{\tcoini}$ will have identical node labels. Hence,
\begin{equation*}
  \dist{\app{\tcontext}{\tcoini}}
  = \nicefrac{1}{2} ⋅ \unit(\node(\true, \unit{\true} ∷ ε)) + \nicefrac{1}{2} ⋅ \unit(\node(\false, \unit{\false} ∷ ε))
\end{equation*}
On the other hand, we have
\begin{align*}
  \app{\tcontext}{\tcoinii}
   &  & \mathrel{\kl[\semeq]{=}}^{\Sampling} &  &
  \tlet*[t]{b}{\tshrink{\tpair{\ttrue, \tfalse}}}{\tif{b}{\tcoinii}{\tcoini}}
   &  & \mathrel{\kl[\semeq]{=}}^{\Sampling} &  &
  \tlet*[t]{b}{\tshrink{\tpair{\ttrue, \tfalse}}}{
    \tif*{b}{
      \tlet*{x}{\tflip{\nicefrac{1}{2}}}{\treturn{x}}
    }{
      \tlet+{y}{\tflip{\nicefrac{1}{2}}}{
        \tlet+{z}{\tflip{\nicefrac{1}{2}}}{
          \quad \treturn{y \teq z}.}}
    }}
\end{align*}
During shrinking, the values of $x$ and $y$ will be identical (as they obtain their values from the same part of the "sample"), but there is only a 50\% chance that $y \teq z$. Hence,
\begin{alignat*}{3}
  \dist{\app{\tcontext}{\tcoinii}}
   & = \nicefrac{1}{4} ⋅ \unit(\node(\true, \unit{\true} ∷ ε))   &
   & + \nicefrac{1}{4} ⋅ \unit(\node(\true, \unit{\false} ∷ ε))    \\
   & + \nicefrac{1}{4} ⋅ \unit(\node(\false, \unit{\true} ∷ ε))  &
   & + \nicefrac{1}{4} ⋅ \unit(\node(\false, \unit{\false} ∷ ε))
\end{alignat*}
If distribution semantics can be represented as an interpretation, then $\dist{\tcoini} = \dist{\tcoinii}$ implies $\dist{\tcontext{\tcoini}} = \dist{\tcontext{\tcoinii}}$. Hence, no such interpretation exists.

\begin{proposition}[Non-Compositionality of Distribution Semantics] \label[proposition]{prop:non-compositionality}
  \AP
  There is no interpretation $F$ such that $\tmsem{t}^F = \dist{t}$ whenever $\typed{t}{\Gen{\Bool}}$.
\end{proposition}

\subsection{Contextual Equivalence}

\AP
The "sampling interpretation" given in \Cref{def:sampling-interpretation} induces a "compositional" semantics, but does not accurately model user-level reasoning and is too fine-grained to be useful for generator equivalence proofs (\Cref{ex:coin-semantics}). Conversely, the "distribution semantics" given in \Cref{def:distribution} has neither of these drawbacks but is non-"compositional". Can we obtain a coarse-grained "interpretation" that more closely models user-level reasoning? To answer this question, we first define what it means for an "interpretation" to be \emph{correct}.

\AP
In the denotational approach to semantics, two standard tasks are to show that a particular denotational semantics is "sound" and "complete" with respect to an existing operational semantics \citep{semantics}. A denotational semantics is "sound" if any two denotationally equal programs are "contextually equivalent", i.e., they produce the same result in any "context" under the operational semantics. "Soundness" rules out ``bad'' "interpretations" (such as the one where all "terms" are equal) and ensures that all conclusions made using the denotational semantics are valid in the operational semantics. Conversely, a denotational semantics is complete if any two "contextually equivalent" programs are denotationally equal. "Completeness" rules out excessively fine "interpretations" (such as one where "terms" are equal iff they are syntactically equal) and ensures that all conclusions that can be made using the operational semantics can also be made using the denotational semantics (i.e., there are no ``blind spots'' in the reasoning power of the denotational semantics). In the context of \hedgehog{}, we consider "distribution semantics" as an operational semantics, and define "contextual equivalence" accordingly. Then, a correct (i.e., "sound" and "complete") "interpretation" is one that coincides with "contextual equivalence".

\begin{definition}[Contexts and Contextual Equivalence] \label[definition]{def:contextual-equivalence}
  \AP
  \phantomintro\hole
  The set of ""contexts"" $\Ctx$ is generated by
  \begin{align*}
    C
    ⩴ {} & \reintro*\hole
    | \reintro*\tabsurd{C}
    | \reintro*\tinl{C}
    | \reintro*\tinr{C}
    | \reintro*\tmatch{C}{x}{t}{x}{t}      \\
    {}
    | {} & \reintro*\tmatch{t}{x}{C}{x}{t}
    | \reintro*\tmatch{t}{x}{t}{x}{C}      \\
    {}
    | {} & \reintro*\tpair{C, t}
    | \reintro*\tpair{t, C}
    | \reintro*\tfst{C}
    | \reintro*\tsnd{C}
    | \reintro*\tlam{x}{τ}{C}
    | \app{C}{t}
    | \app{t}{C}
    | \reintro*\treturn{C}                 \\
    {}
    | {} & \reintro*\tlet{x}{C}{t}
    | \reintro*\tlet{x}{t}{C}.
  \end{align*}
  The ""substitution"" $\subst{C}{t}$ of a "term" $t$ into a context $C$ is obtained by substituting the (unique) instance of $\hole$ in $C$ for $t$. We write $\intro*\ctxtyped{C}{Γ}{τ}{Γ'}{τ'}$ when $\typed[Γ']{C[t]}{τ'}$ for all "terms" $\typed[Γ]{t}{τ}$. Two "terms" $\typed[Γ]{t₁, t₂}{τ}$ are ""contextually equivalent"" (written $\disteq[Γ]{t₁}{t₂}[τ]$) if $\dist{\subst{C}{t₁}} = \dist{\subst{C}{t₂}}$ for all contexts $\ctxtyped{C}{Γ}{τ}{\sing}{\tau'}$, where $τ' ∈ \Type$. We write $\disteq{t₁}{t₂}[τ]$ in place of $\disteq[\sing]{t₁}{t₂}[τ]$ and omit the "type" ascription ($: τ$) when it can be inferred from the surrounding context.
\end{definition}

\AP
Intuitively, \Cref{def:contextual-equivalence} states that two terms are contextually equivalent if one can be substituted for the other in any context without changing the overall "distribution" of the term. For example, we refer to \Cref{ex:coin-distribution}, from which we conclude $\tcoini$ and $\tcoinii$ are \emph{not} contextually equivalent: they produce different "distributions" in the context $\tcontext{\hole}$.

\AP
We now define our correctness criteria for "interpretations".

\begin{definition} \label[definition]{def:soundness-and-completeness}
  \AP
  An "interpretation" $F$ is
  \begin{enumerate*}
    \item ""sound"" if $\semeq{F}[Γ]{t₁}{t₂}[τ] \implies \disteq[Γ]{t₁}{t₂}[τ]$, and
    \item ""complete"" if $\disteq[Γ]{t₁}{t₂}[τ] \implies \semeq{F}[Γ]{t₁}{t₂}[τ]$.
  \end{enumerate*}
\end{definition}

\AP
The "sampling interpretation" $\Sampling$ is trivially sound (generators which represent the same "random variables" have the same "distribution"). Surprisingly, $\Sampling$ is also complete! We sketch our "completeness" proof for terminating generator terms with no free variables.

\AP
We proceed by contraposition. In essence, the problems we have observed in \Cref{ex:coin-syntax,ex:coin-semantics,ex:coin-distribution} generalize to terms other than $\tcoini$ and $\tcoinii$. Suppose we have two terms $t_1, t_2 : \Gen{τ}$ such that $t₁ \mathrel{\not\mathrel{\kl[\semeq]{=}}}^{\Sampling} t₂$. Then there exists some $σ ∈ \Sample$ such that $\tmsem{t₁}^{\Sampling}{σ} ≠ \tmsem{t₂}^{\Sampling}{σ}$. Define
\begin{equation*} \label{eq:pathological}
  \begin{array}{@{}l@{}}
    \ctxtyped{\intro*\tdiscrim}{\sing}{\Gen{τ}}{\sing}{\Gen{τ}} \\
    \tdiscrim =
    \tlet{b}{\tshrink{\tpair{\ttrue, \tfalse}}}{
      \tif{b}{\hole}{t₁}.}
  \end{array}
\end{equation*}
First, we consider the case of $\subst{\tdiscrim}{t₁}$:
\begin{align*}
  \subst{\tdiscrim}{t₁}
   & \mathrel{\kl[\semeq]{=}}^{\Sampling} \tlet{b}{\tshrink{\tpair{\ttrue, \tfalse}}}{\tif{b}{t₁}{t₁}} \\
   & \mathrel{\kl[\semeq]{=}}^{\Sampling} \tlet{b}{\tshrink{\tpair{\ttrue, \tfalse}}}{t₁}
\end{align*}
The first shrink produced by $\subst{\tdiscrim}{t₁}$ is always the same element returned during the testing phase, as $b$ first shrinks from $\ttrue$ to $\tfalse$. More concretely, for any $σ' ∈ \Sample$,
\begin{equation}
  \tmsem{\subst{\tdiscrim}{t₁}}^{\Sampling}{σ'} = \node(b, \node(b, \pl{x}) ∷ \pl{x}) \tag{$†$} \label{eq:context-form}
\end{equation}
where $\node(b, \pl{x}) = \tmsem{t₁}^{\Sampling} (\projr{σ'})$. Now consider the case of $\subst{\tdiscrim}{t₂}$:
\begin{equation*}
  \subst{\tdiscrim}{t₂}
  \mathrel{\kl[\semeq]{=}}^{\Sampling} \tlet{b}{\tshrink{\tpair{\ttrue, \tfalse}}}{\tif{b}{t₂}{t₁}}.
\end{equation*}
Without loss of generalization, assume $σ = \projr{σ'}$ for some $σ' ∈ \Sample$. Then
\begin{equation*}
  \tmsem{\subst{\tdiscrim}{t₂}}^{\Sampling}{σ'} = \node(a, \node(b, \pl{y}) ∷ \pl{x})
\end{equation*}
where $(a, \pl{x}) = \tmsem{t₂}^{\Sampling}σ$ and $(b, \pl{y}) = \tmsem{t₁}^{\Sampling}σ$. Since $\tmsem{t_1}^{\Sampling} σ ≠ \tmsem{t_2}^{\Sampling} σ$, either $a ≠ b$ or $\pl{x} ≠ \pl{y}$. In either case, the result of $\tmsem{\subst{\tdiscrim}{t_1}}^{\Sampling}{σ''}$ is \emph{not} of the form described in (\ref{eq:context-form}). Thus $\node(a, \node(b, \pl{y}) ∷ \pl{x})$ occurs with positive probability under $\dist{\subst{\tdiscrim}{t_2}}$ but occurs with zero probability under $\dist{\subst{\tdiscrim}{t_1}}$.

\begin{theorem} \label[theorem]{thm:sampling-correct}
  \AP
  The "sampling interpretation" $\Sampling$ is "sound" and "complete".
\end{theorem}

\ifextended
  \begin{proof}
    \emph{Soundness.}
    Suppose $\semeq{\Sampling}[Γ]{t₁}{t₂}[τ]$ and $\ctxtyped{C}{Γ}{τ}{\sing}{\Bool}$.
    Since $\Sampling$ is (by definition) compositional, we have $\semeq{\Sampling}[Γ]{\subst{C}{t₁}}{\subst{C}{t₂}}[τ]$.
    Hence $\dist{\subst{C}{t₁}} = \dist{\subst{C}{t₂}}$.

    \emph{Completeness.}
    Let
    \[
      \app{P}{τ} ≡ ∀ \typed[Γ]{t₁, t₂}{τ}, \disteq[Γ]{t₁}{t₂} \implies \semeq{\Sampling}[Γ]{t₁}{t₂}[τ].
    \]
    Given $τ$, we prove $\app{P}{τ}$ by structural induction on $τ$.
    \begin{itemize}
      \item \emph{Cases $τ = \Empty$ and $τ = \Unit$}. These cases are trivial by the fact that all elements of these types are equal.
      \item \emph{Case $τ = \Prod{τ₁}{τ₂}$}. Suppose $\disteq[Γ]{t₁}{t₂}[\Prod{τ₁}{τ₂}]$, where $\typed[Γ]{t₁, t₂}{\Prod{τ₁}{τ₂}}$, and suppose the inductive hypotheses $\app{P}{τ₁}$ and $\app{P}{τ₂}$. It suffices to show $\semeq{\Sampling}[Γ]{\tfst{t₁}}{\tfst{t₂}}[τ₁]$ and $\semeq{\Sampling}[Γ]{\tsnd{t₁}}{\tsnd{t₂}}[τ₂]$, which, by the inductive hypotheses, reduces to showing $\disteq[Γ]{\tfst{t₁}}{\tfst{t₂}}[τ₁]$ and $\disteq[Γ]{\tsnd{t₁}}{\tsnd{t₂}}[τ₂]$, which follows from $\disteq[Γ]{t₁}{t₂}[\Prod{τ₁}{τ₂}]$.
      \item \emph{Case $τ = \Fun{τ₁}{τ₂}$}. Suppose $\disteq[Γ]{t₁}{t₂}[\Fun{τ₁}{τ₂}]$, where $\typed[Γ]{t₁, t₂}{\Fun{τ₁}{τ₂}}$, and suppose the inductive hypotheses $\app{P}{τ₁}$ and $\app{P}{τ₂}$. It suffices to show $\semeq{\Sampling}[Γ]{\app{t₁}{t₃}}{\app{t₂}{t₃}}[τ₂]$  for all $\typed[Γ]{t₃}{τ₁}$. By the inductive hypothesis, it suffices to show $\disteq[Γ]{\app{t₁}{t₃}}{\app{t₂}{t₃}}[τ₂]$, which follows from $\disteq[Γ]{t₁}{t₂}[\Fun{τ₁}{τ₂}]$.
      \item \emph{Case $τ = \Sum{τ₁}{τ₂}$}. Suppose $\disteq[Γ]{t₁}{t₂}[\Sum{τ₁}{τ₂}]$, where $\typed[Γ]{t₁, t₂}{\Sum{τ₁}{τ₂}}$, and suppose the inductive hypotheses $\app{P}{τ₁}$ and $\app{P}{τ₂}$. We consider, by cases, whether there exists terms $t₃ : τ₁$ and $t₄ : τ₂$.
            \begin{itemize}
              \item \emph{Case $t₃$ exists and $t₄$ exists}. It suffices to show
                    \begin{gather*}
                      \semeq{\Sampling}[Γ]{\tmatch{t₁}{x}{x}{x}{t₃}}{\tmatch{t₂}{x}{x}{x}{t₃}}[τ₁] \\
                      \semeq{\Sampling}[Γ]{\tmatch{t₁}{x}{t₄}{x}{x}}{\tmatch{t₂}{x}{t₄}{x}{x}}[τ₂],
                    \end{gather*}
                    which, by the inductive hypotheses, reduce to
                    \begin{gather*}
                      \disteq[Γ]{\tmatch{t₁}{x}{x}{x}{t₃}}{\tmatch{t₂}{x}{x}{x}{t₃}}[τ₁] \\
                      \disteq[Γ]{\tmatch{t₁}{x}{t₄}{x}{x}}{\tmatch{t₂}{x}{t₄}{x}{x}}[τ₂],
                    \end{gather*}
                    which both follow from $\disteq[Γ]{t₁}{t₂}[\Sum{τ₁}{τ₂}]$.
              \item \emph{Case $t₃$ exists and $t₄$ does not exist.} If $t₄$ does not exist, then $τ₂$ is isomorphic to $\Empty$, and hence there is an alternate term $\typed[x : τ₂]{t₄'}{τ₁}$ involving $\tabsurd$. It suffices to show
                    \begin{gather*}
                      \semeq{\Sampling}[Γ]{\tmatch{t₁}{x}{x}{x}{t₃}}{\tmatch{t₂}{x}{x}{x}{t₃}}[τ₁] \\
                      \semeq{\Sampling}[Γ]{\tmatch{t₁}{x}{t₄'}{x}{x}}{\tmatch{t₂}{x}{t₄'}{x}{x}}[τ₂],
                    \end{gather*}
                    which, by the inductive hypotheses, reduce to
                    \begin{gather*}
                      \disteq[Γ]{\tmatch{t₁}{x}{x}{x}{t₃}}{\tmatch{t₂}{x}{x}{x}{t₃}}[τ₁] \\
                      \disteq[Γ]{\tmatch{t₁}{x}{t₄'}{x}{x}}{\tmatch{t₂}{x}{t₄'}{x}{x}}[τ₂],
                    \end{gather*}
                    which both follow from $\disteq[Γ]{t₁}{t₂}[\Sum{τ₁}{τ₂}]$.
              \item \emph{Case $t₃$ does not and $t₄$ exists.} Symmetric to the previous case.
              \item \emph{Case $t₃$ does not and $t₄$ does not exist.} In this case, $τ₁$ and $τ₂$ are isomorphic to $\Empty$, and hence $\Sum{τ₁}{τ₂}$ is isomorphic to $\Empty$. Thus $\semeq{\Sampling}[Γ]{t₁}{t₂}[\Empty]$ trivially.
            \end{itemize}
      \item \emph{Case $τ = \Gen{τ'}$}. We proceed by contraposition. Suppose $Γ ⊢ t₁ \mathrel{\not\mathrel{\kl[\semeq]{=}}}^{\Sampling} t₂$. Then there exists some $ρ ∈ \tysem{Γ}$ and $σ ∈ \Sample$ such that $\tmsem{t₁}^{\Sampling}{ρ}{σ} ≠ \tmsem{t₂}^{\Sampling}{ρ}{σ}$. We proceed using $\tdiscrim$ as described above. Thus
            \begin{enumerate}
              \item for all $ρ$ and $σ$, $\tmsem{\subst{\tdiscrim}{t₂}}^{\Sampling}{ρ₁}{σ₁} = (a, \node(a, \pl{x}), \pl{x})$ for some $a$ and $\pl{x}$, and
              \item $\tmsem{\subst{\tdiscrim}{t₂}}^{\Sampling}{ρ₂}{σ₂} = (a, \node(b, \pl{y}), \pl{x})$ for some $ρ$, $σ$, $a ≠ b$ and $\pl{y} ≠ \pl{x}$.
            \end{enumerate}
            Since our terms represent \emph{terminating} programs, the set of choices of $σ₂$ must be infinite and of non-zero measure. Thus $\dist{\subst{\tdiscrim}{t₁}} = \dist{\subst{\tdiscrim}{t₂}}$, which completes the proof.
    \end{itemize}
  \end{proof}
\fi

\begin{corollary} \label[corollary]{corr:sampling-correct}
  \AP
  Let $F$ be a "sound" and "complete" "interpretation". Then two "terms" $\typed[Γ]{t₁, t₂}{τ}$ are "semantically equivalent" under $F$ iff they are "semantically equivalent" under $\Sampling$, i.e.,
  \[
    \semeq{F}[Γ]{t₁}{t₂} \iff \disteq[Γ]{t₁}{t₂}.
  \]
\end{corollary}

\subsection{Summary}

\AP
We have formalized \hedgehog{}'s syntax as well as its \kl[sampling semantics]{sampling} and "distribution semantics" with an eye towards proving program optimizations sound. The "sampling semantics" is "compositional" but unusably fine-grained: many common program optimizations are unsound under the "sampling semantics" (\Cref{thm:unsound-transformations}). On the other hand, the "distribution semantics" is coarse-grained but non-"compositional". By \Cref{corr:sampling-correct}, \emph{every} ("sound" and "complete") "compositional" semantics must identify the exact same set of "terms" as the "sampling semantics". Thus, none of the transformations discussed in \Cref{thm:unsound-transformations} are permitted under \emph{any} "sound" "interpretation" of \hedgehog{}, which renders \hedgehog{} a poor target for optimization. We stress that \Cref{thm:sampling-correct} is a consequence of combining probabilistic and shrinking effects: if either $\tflip$ or $\tshrink$ are removed from \hedgehog{}, then \Cref{thm:sampling-correct} no longer holds.

\end{scope}

%% file: section/4-solution.tex
\begin{scope}
\knowledgeimport{background}
\knowledgeimport{hedgehog*}

\section{\hedgehog*{}} \label{sec:solution}

\begin{outline}[b]
  \textbf{Solution}
  \begin{itemize}
    \item Syntax
          \begin{itemize}
            \item Types and Terms
            \item Typing Relation
            \item Coin Examples
          \end{itemize}
    \item Semantics
          \begin{itemize}
            \item Interpretation and Semantics
            \item Sampling Interpretation
          \end{itemize}
    \item Distributions
          \begin{itemize}
            \item Distribution Semantics
            \item Distribution Interpretation (compositionality)
          \end{itemize}
    \item Translation
          \begin{itemize}
            \item Translation Definition
            \item Substitution Lemma
          \end{itemize}
  \end{itemize}
\end{outline}
\AP
In \Cref{sec:analysis}, we demonstrated that \hedgehog{} is unsuited for performing generator optimizations. In particular, any two programs with identical distributions can be distinguished within a pathological context, $\tdiscrim$, unless they have identical sampling semantics. To overcome this flaw, we propose to instead work with \emph{\hedgehog*{}}, a restricted version of \hedgehog{} based on the arrow calculus \citep{arrow-calculus} which rules out such pathological contexts.

\AP
Intuitively, $\tdiscrim$ exploits hidden statistical dependence between "generators@@hedgehog" in different branches. This dependence is a direct result of both sides of a $\tmatchkeyword$ term being evaluated with the same "sample". Thus, \hedgehog*{} rules out this behaviour by ensuring that each branch of a case analysis ($\tmatchkeyword$) receives different "samples". This is achieved by
\begin{enumerate*}
  \item redefining the semantics of case analysis and
  \item adding additional type system restrictions.
\end{enumerate*}

\AP
In this section, we define the \hedgehog*{} language. This section is structurally similar to \Cref{sec:analysis}. We first define \hedgehog*{}'s syntax and type system as a modification of \hedgehog{}'s (\Cref{def:syntax*}, \Cref{def:typing*}). We then define a "sampling semantics" for \hedgehog*{} (\Cref{def:sampling-interpretation*}) and derive from it a "distribution semantics" (\Cref{def:distribution*}). Unlike \hedgehog{}, \hedgehog*{}'s "distribution semantics" is "compositional" (we show it in \Cref{thm:compositionality}).

\subsection{Syntax}

\begin{definition}[\hedgehog*{} Types and Terms] \label[definition]{def:syntax*}
  \AP
  The set of \hedgehog*{} ""types"" (denoted by $\Type*$) is defined by the following modifications to $\Type$ (\Cref{def:syntax}):
  \begin{equation*}
    τ ⩴ … | \cancel{\Gen{τ}} | \intro*\GFun{τ}{τ}
  \end{equation*}
  The set of \hedgehog*{} ""terms"" (denoted by $\Term*$) is defined identically to $\Term$ (\Cref{def:syntax}), but types are drawn from $\Type*$ instead of $\Type$ (e.g., in $\tlam{x}{τ}{t}$).
\end{definition}

\AP
\hedgehog*{} replaces the generator type $\Gen{τ}$ with a "generator function type" $\GFun{τ₁}{τ₂}$. Intuitively, an element of $\GFun{τ₁}{τ₂}$ is a function that
\begin{enumerate*}
  \item takes in an element of type $τ₁$,
  \item performs some probabilistic and shrinking effects, and then
  \item returns a value of type $τ₂$.
\end{enumerate*}
In other words, $\GFun{τ₁}{τ₂}$ is analogous to $\Fun{τ₁}{\Gen{τ₂}}$. At the "term" level, \hedgehog*{} reuses the existing "function application@@hedgehog" ($\app{t₁}{t₂}$) and "anonymous function@@hedgehog" constructs ($\tlam{x}{\tau}{t}$) for "generator functions".

\AP
The most significant changes are in the "typing" rules (\Cref{def:typing*}).

\begin{definition}[\hedgehog*{} Typing] \label[definition]{def:typing*}
  \AP
  The set of ""environments"" $\Env*$ and the ""typing relation"" $(\typed*[\blank]{\blank}{\blank}) ⊆ \Env* × \Term* × \Type*$ are defined analogously to \Cref{def:typing}, but we remove the rules for $\tletkeyword$, $\treturn$, $\tflip$, and $\tshrink$, and add
  \begin{mathpar}
    \infer{
      \ctyped{Γ}{x : τ₁}{t}{τ₂}
    }{\typed*[Γ]{\tlam{x}{τ₁}{t}}{\GFun{τ₁}{τ₂}}}

    \infer{
    }{\typed*[Γ]{\tflip{p}}{\GFun{\Unit}{\Bool}}}

    \infer{
    }{\typed*[Γ]{\tshrink}{\GFun{τ × ⋯ × τ}{τ}}}
  \end{mathpar}
  where the ""command typing relation"" $(\ctyped{\blank}{\blank}{\blank}{\blank}) ⊆ \Env* × \Env* × \Term* × \Type*$ is mutually defined to be the smallest relation satisfying
  \begin{mathpar}
    \infer{
      \typed*[Γ \join Δ]{t}{τ}
    }{\ctyped{Γ}{Δ}{\treturn{t}}{τ}}

    \infer{
      x ∉ \Dom{Δ} \\
      \ctyped{Γ}{Δ}{t₁}{τ₁} \\
      \ctyped{Γ}{Δ \join x : τ₁}{t₂}{τ₂}
    }{\ctyped{Γ}{Δ}{\tlet{x}{t₁}{t₂}}{τ₂}}

    \infer{
      \typed*[Γ]{t₁}{\GFun{τ₁}{τ₂}} \\
      \ctyped{Γ}{Δ}{t₂}{τ₁}
    }{\ctyped{Γ}{Δ}{\app{t₁}{t₂}}{τ₂}}

    \infer{
      x₁, x₂ ∉ \Dom{Δ} \\
      \typed*[Γ \join Δ]{t₀}{\Sum{τ₁}{τ₂}} \\\\
      \ctyped{Γ}{Δ \join x₁ : τ₁}{t₁}{τ₃} \\
      \ctyped{Γ}{Δ \join x₂ : τ₂}{t₂}{τ₃}
    }{\ctyped{Γ}{Δ}{\tmatch{t₀}{x₁}{t₁}{x₂}{t₂}}{τ₃}}.
  \end{mathpar}
\end{definition}

\AP
\hedgehog*{}'s "typing" rules divide terms into two classes: normal terms (satisfying $\typed*[\blank]{\blank}{\blank}$) and ""command terms"" (satisfying $\ctyped{\blank}{\blank}{\blank}{\blank}$). Command terms only occur inside an "anonymous function@@hedgehog" with a "generator function type". Command terms include effect-free computations ($\treturn$), sequential composition ($\tletkeyword$), effectful function application ($\app{t₁}{t₂}$), and case analyses which result in commands ($\tmatchkeyword$). The "typing" rules for command terms have two contexts: one for normal values ($Γ$) and one for values produced by the result of another command or given as an input to the function ($Δ$). In an effectful function application $\app{t₁}{t₂}$, only the argument term $t₂$ can refer to any variable in $Δ$. This rules out command terms such as
\[
  \tlet{f}{t₁}{\tlet{x}{t₂}{\app{f}{x},}}
\]
where the effects performed by the last term ($\app{f}{x}$) are determined by the (arbitrarily complex) behaviour of $t₁$, i.e. the control flow of a command term is \emph{fixed}.

\begin{example} \label[example]{ex:coin-syntax*}
  \AP
  The "terms" $\tcoini$ and $\tcoinii$ (\Cref{ex:coin-syntax}) are still typed in \hedgehog*{}, but they are "command terms" instead of normal "terms", i.e.,
  \begin{align*}
     & \ctyped{Γ}{Δ}{\tcoini}{\Bool}   &
     & \ctyped{Γ}{Δ}{\tcoinii}{\Bool}.
  \end{align*}
  \AP
  \phantomintro\tcoinix
  \phantomintro\tcoiniix
  To form normal "terms", we place each "term" inside an "anonymous function@@hedgehog":
  \begin{align*}
    \begin{array}[t]{@{}l@{}}
      \tcoinix : \GFun{\Unit}{\Bool} \\
      \reintro*\tcoinix = \tlam{u}{\Unit}{\tlet+{x}{\tflip{\nicefrac{1}{2}}}{\tlet+{y}{\tflip{\nicefrac{1}{2}}}{\quad \treturn{x \teq y}}}}
    \end{array}
     &  &
    \begin{array}[t]{@{}l@{}}
      \tcoiniix : \GFun{\Unit}{\Bool} \\
      \reintro*\tcoiniix = \tlam{u}{\Unit}{\tlet+{x}{\tflip{\nicefrac{1}{2}}}{\treturn{x}.}}
    \end{array}
  \end{align*}
  \AP
  \phantomintro\tcontextx
  We similarly wrap the function $\tcontext$ as
  \[
    \begin{array}[t]{@{}l@{}}
      \tcontextx : \Fun{(\GFun{\Unit}{\Bool})}{(\GFun{\Unit}{\Bool})} \\
      \reintro*\tcontextx =
      \tlam{f}{\GFun{\Unit}{\Bool}}{
        \tlam{u}{\Unit}{
          \tlet+{b}{\tshrink{\tpair{\ttrue, \tfalse}}}{
            \tif{b}{\app{f}{\tunit}}{\app{\tcoinix}{\tunit}.}}}}
    \end{array}
  \]
\end{example}

\subsection{Semantics}

\begin{definition}[\hedgehog*{} Interpretation and Semantics] \label[definition]{def:semantics*}
  \AP
  A \hedgehog*{} ""interpretation"" $F : \Qbs × \Qbs → \Qbs$ is a "pre-arrow" equipped with additional operations
  \begin{align*}
     & \reintro*\flip_F : [0, 1] \qfun \app{F}(\one, \B)            &
     & \reintro*\shrink^n_F ∈ \app{F}(\textstyle ∏_{i = 0}^n A, A).
  \end{align*}
  An interpretation induces a ""type and environment semantics"" $\tysem*{\blank}^F : \Type* ∪ \Env* → \Qbs$, defined analogously to \Cref{def:semantics}, but with
  \[
    \tysem*{\GFun{τ₁}{τ₂}}^F = \app{F}(\tysem*{τ₁}^F, \tysem*{τ₂}^F).
  \]
  An interpretation also induces, for all $\typed*[Γ]{t}{τ}$, a ""term semantics"" $\tmsem*[Γ]{t}[τ]^F : \tysem{Γ}^F \qfun \tysem{τ}^F$ defined analogously to \Cref{def:semantics} with
  \begin{mathpar}
    \tmsem*[Γ]{\tlam{x}{τ₁}{t}}[\GFun{τ₁}{τ₂}]^F = \tmsem*[Γ \semi x : τ₁]{t}[τ₂]^F

    \tmsem*[Γ]{\tflip{p}}[\GFun{\Unit}{\Bool}]^F{ρ} = \flip_F{p}

    \tmsem*[Γ]{\tshrink}[\GFun{\underbrace{\Prod{τ}{\Prod{⋯}{τ}}}_{n + 1\text{ times}}}{\tau}]^F{ρ} = \shrink^n_F
  \end{mathpar}
  where the function $\tmsem*[Γ \semi Δ]{t}[τ]^F : \tysem*{Γ}^F \qfun \app{F}(\tysem*{Δ}^F, \tysem*{τ}^F)$ is defined for all $\ctyped{Γ}{Δ}{t}{τ}$ in \Cref{fig:semantics*}.
\end{definition}
\begin{figure}[t]
  \begin{mathpar}
    \tmsem*[Γ \semi Δ]{\treturn{x}}[τ]^F{ρ} = \arr(ρ' ↦ \tmsem*[Γ \join Δ]{x}[τ]^F(ρ \join ρ'))

    \begin{array}{@{}l@{}}
      \tmsem*[Γ \semi Δ]{\tlet{x}{t₁}{t₂}}[τ₂]^F{ρ}                    \\
      \qquad = \tmsem*[Γ \semi Δ \join x : τ₁]{t₂}[τ₂]^F{ρ}            \\
      \qquad {} \acomp \arr((ρ', ρ'_x) ↦ (ρ'_y)_{y ∈ \Dom{Δ} ∪ \{x\}}) \\
      \qquad {} \acomp \afirst(\tmsem*[Γ \semi Δ]{t₁}[τ₁]^F{ρ})        \\
      \qquad {} \acomp \arr(ρ' ↦ (ρ', ρ'))                             \\
      \quad \text{where } \ctyped{Γ}{Δ}{t₁}{τ₁}, \quad \ctyped{Γ}{Δ \join x : τ₁}{t₂}{τ₂}
    \end{array}

    \begin{array}{@{}l@{}}
      \tmsem*[Γ \semi Δ]{\app{t₁}{t₂}}[τ₂]^F{ρ} = \tmsem*[Γ]{t₁}[\GFun{τ₁}{τ₂}]^F{ρ} \acomp \tmsem*[Γ \semi Δ]{t₂}[τ₁]^F{ρ} \\
      \quad \text{where } \typed*[Γ]{t₁}{\GFun{τ₁}{τ₂}},\quad \ctyped{Γ}{Δ}{t₂}{τ₁}
    \end{array}

    \begin{array}{@{}l@{}}
      \tmsem*[Γ \semi Δ]{\tmatch{t₀}{x₁}{t₁}{x₂}{t₂}}[τ₃]^F{ρ}                                   \\
      \qquad = \arr((a, b) ↦ (b, a))                                                             \\
      \qquad {} \acomp {} \aleft(\tmsem*[Γ \semi Δ \join x_2 : τ₂]{t₂}[τ₃]^F{ρ})                 \\
      \qquad {} \acomp {} \arr((a, b) ↦ (b, a))                                                  \\
      \qquad {} \acomp {} \aleft(\tmsem*[Γ \semi Δ \join x_1 : τ₁]{t₁}[τ₃]^F{ρ})                 \\
      \qquad {} \acomp {} \arr(ρ' ↦ \distrib(ρ', \tmsem*[Γ \join Δ]{t₀}[τ₁ + τ₂]^F(ρ \join ρ'))) \\
      \quad \text{where } \typed*[Γ \join Δ]{t₀}{τ₁ + τ₂}, \quad \ctyped{Γ}{Δ \join x₁ : τ₁}{t₁}{τ₃}, \quad \ctyped{Γ}{Δ \join x₂ : τ₂}{t₂}{τ₃}
    \end{array}
  \end{mathpar}
  \vspace{-0.2in}
  \Description{Hedgehog* Command Semantics}
  \caption{\hedgehog*{} Command Semantics.}
  \label{fig:semantics*}
  \vspace{-0.15in}
\end{figure}

\AP
The semantics for "terms" are adapted from \citet{arrow-calculus}. The semantics of \hedgehog*{} "terms" are defined analogously to \hedgehog{}, where an "interpretation" provides implementations of "command terms". One significant difference is that the semantics of $\tmatchkeyword$ (specifically within a "command term") is now also defined by the "interpretation". As with \Cref{def:semantics}, we only require an "interpretation" to be a "pre-arrow" and not an "arrow", since our "sampling interpretation" is not an "arrow".

\begin{definition}[\hedgehog*{} Sampling Interpretation] \label[definition]{def:sampling-interpretation*}
  \AP
  The ""sampling interpretation"" $\Sampling* : \Qbs \times \Qbs → \Qbs$ of \hedgehog*{} is defined in \Cref{fig:sampling-interpretation*}.
\end{definition}
\begin{figure}[t]
  \begin{mathpar}
    \Sampling*(A, B)
    = (\Sample × A \qfun \Tree{B})

    \arr_{\Sampling*}{f}(σ, a) = \unit(\app{f}{a})

    \app(f \acomp_{\Sampling*} g)(σ, a) = \bind(b ↦ \app{f}(\app{πᵣ}{σ}, b))(\app{g}(\app{πₗ}{σ}, a))

    \afirst_{\Sampling*}{f}(σ, (a, c)) = \map(\blank, c)(\app{f}(σ, a))

    \aleft_{\Sampling*}{f}(σ, \app{ι₁}{a}) = \map{ι₁}(\app{f}(σ, a))

    \aleft_{\Sampling*}{f}(σ, \app{ι₂}{b}) = \unit(\app{ι₂}{b})

    \flip_{\Sampling*}{p}(σ, \sing) = \unit(σ < p)

    \shrink^n_{\Sampling*}(σ, (a₀, a₁, …, aₙ)) = \node(a₀, \unit{a₁} ∷ … ∷ \unit{aₙ} ∷ ε)
  \end{mathpar}
  \Description{Hedgehog*'s Sampling Interpretation}
  \vspace{-0.2in}
  \caption{\hedgehog*{}'s Sampling Interpretation.}
  \label{fig:sampling-interpretation*}
  \vspace{-0.15in}
\end{figure}

\AP
\hedgehog{}'s "sampling semantics" defines generators as "random variables", and \hedgehog*{} analogously defines "generator functions" as "random functions". The remaining definitions also mirror \Cref{def:semantics}, e.g., where composition $f \acomp_{\Sampling*} g$ splits its "sample" between arguments $f$ and $g$.

\begin{remark} \label[remark]{rem:restrictions}
  \emph{Only} redefining the semantics of $\tmatchkeyword$ is not sufficient to rule out pathological contexts such as $\tdiscrim$. \hedgehog*{}'s additional type system restrictions are necessary due to the presence of higher-order computation. For example, recall from \Cref{prop:non-compositionality},
  \[
    \tcontext{\tcoinii} \mathrel{\kl[\semeq]{=}}^{\Sampling} \tlet{b}{\tshrink{\tpair{\ttrue, \tfalse}}}{\tif{b}{\tcoini}{\tcoinii}}.
  \]
  This can be expressed equivalently --- without any case analysis --- as
  \[
    \tlet{m}{\tshrink{\tpair{\tcoini, \tcoinii}}}{m}.
  \]
  The \hedgehog*{} version of this term is
  \[
    \tlet{m}{\tshrink{\tpair{\tcoinix, \tcoiniix}}}{\app{m}{\tunit}}
  \]
  which does not satisfy the "typing relation" because $\tletkeyword$-bound variables (i.e., $m$), cannot appear on the left hand side of a function application (as in $\app{m}{\tunit}$).
\end{remark}

\begin{example} \label[example]{ex:coin-semantics*}
  \AP
  The "sampling semantics" of $\tcoinix, \tcoiniix : \GFun{\Unit}{\Bool}$ (\Cref{ex:coin-syntax*}) are nearly identical to their counterparts in \Cref{ex:coin-semantics}:
  \begin{align*}
    \tmsem*{\tcoinix}^{\Sampling*} (σ, u)  & = \unit(σ₁ < \nicefrac{1}{2} ⇔ σ₂ < \nicefrac{1}{2}) &
    \tmsem*{\tcoiniix}^{\Sampling*} (σ, u) & = \unit(σ₁ < \nicefrac{1}{2})
  \end{align*}
  where $σ₁, σ₂ ∈ \Sample$ obtained from $σ$ by some sequence of invocations of $\projl$ and $\projr$. The exact sequences are unimportant, but they are much longer than in \Cref{ex:coin-semantics} because the semantics of $\tletkeyword$ (\Cref{fig:semantics*}) involves a large number of compositions ($\acomp_{\Sampling*}$). \Citet[\S 4.2]{arrow-programming} describes how these redundant compositions can be eliminated.

  On the other hand, the semantics of $\tcontextx{\tcoinix}$ is
  \begin{equation*}
    \tmsem*{\tcontextx{\tcoinix}}^{\Sampling*}{σ} =
    \node(σ₁ < \nicefrac{1}{2} ⇔ σ₂ < \nicefrac{1}{2}, \unit(σ₃ < \nicefrac{1}{2} ⇔ σ₄ < \nicefrac{1}{2}) ∷ ε)
  \end{equation*}
  where $σ₁, σ₂, σ₃, σ₄ ∈ \Sample$ are also obtained from $σ$ by some sequence of invocations of $\projl$ and $\projr$. Unlike \hedgehog{}, \hedgehog*{} ensures that each branch in $\tcontextx$ receives an independent "sample", so $\tcontextx{\tcoinix}$ does not return "trees" where all nodes are identically labelled. This example also demonstrates that the command term $\tif{t}{t'}{t'}$ is generally \emph{not} identical to $t'$, i.e., \Cref{eq:if-elim} does not hold in \hedgehog*{}.
\end{example}

\AP
Like \hedgehog{}, \hedgehog*{} interpretations induce a "semantic equivalence" relation.

\begin{definition}[\hedgehog*{} Semantic Equivalence]
  \AP
  Two \hedgehog*{} terms $\typed*[Γ]{t₁, t₂}{τ}$ are ""semantically equivalent"" with respect to an interpretation $F$ (written $\semeq*{F}[Γ]{t₁}{t₂}[τ]$) if $\tmsem*[Γ]{t₁}[τ]^F = \tmsem*[Γ]{t₂}[τ]^F$. We write $\semeq*{F}{t₁}{t₂}[τ]$ for $\semeq*{F}[\sing]{t₁}{t₂}[τ]$ and omit the type ascription ($: τ$) when it can be inferred from the surrounding context.
\end{definition}


\subsection{Distributions}

\begin{definition}[\hedgehog*{} Distribution Semantics] \label[definition]{def:distribution*}
  \AP
  The ""distribution semantics"" of term $t : \GFun{\Unit}{\Bool}$ is given by the function
  \[
    \begin{array}{@{}l@{}c@{}l@{}}
      \dist*{ & \blank & } : \Term* → \Dist(\one \qfun \B)                    \\
      \dist*{ & t      & } = \map(\curry{\tmsem*{t}^{\Sampling*}}){\uniform}.
    \end{array}
  \]
\end{definition}

\AP
\Cref{def:distribution*} is the \hedgehog*{} analogue of \Cref{def:distribution}.

\begin{example} \label[example]{ex:coin-distribution*}
  \AP
  We compute the "distributions" of $\tcoinix, \tcoiniix : \GFun{\Unit}{\Bool}$ given in \Cref{ex:coin-syntax*}:
  \begin{align*}
    \dist*{\tcoinix}  & = \nicefrac{1}{2} ⋅ \unit(\sing ↦ \unit{\true}) + \nicefrac{1}{2} ⋅ \unit(\sing ↦ \unit{\false}) \\
    \dist*{\tcoiniix} & = \nicefrac{1}{2} ⋅ \unit(\sing ↦ \unit{\true}) + \nicefrac{1}{2} ⋅ \unit(\sing ↦ \unit{\false})
  \end{align*}
  These "distributions" are essentially identical to the ones given in \Cref{ex:coin-distribution}.
\end{example}

\begin{definition}[\hedgehog*{} Distribution Interpretation] \label[definition]{def:distribution-interpretation*}
  \AP
  The ""distribution interpretation"" $\Distrib : \Qbs → \Qbs$ is defined in \Cref{fig:distribution-interpretation*}.
\end{definition}

\begin{figure}[t]
  \begin{mathpar}
    \Distrib(A, B)
    = \Dist(A \qfun \Tree{B})

    \arr_{\Distrib}{μ} = \unit(\unit ∘ μ)

    μ₁ \acomp_{\Distrib} μ₂
    = \bind(a₂ ↦ \map(a₁ ↦ \bind{a₁} ∘ a₂){μ₁}){μ₂}

    \afirst_{\Distrib}{μ} = \map(μ' ↦ (a, c) ↦ \map(\blank, c)(μ' a)){μ}

    \aleft_{\Distrib}{μ} = \map(μ' ↦ \elim{\map{\inj_1} ∘ μ', \inj_2}){μ}

    \flip_{\Distrib}{p}
    = \map((σ, \sing) ↦ σ < p){\uniform}

    \shrink^n_{\Distrib} = \unit((a₀, a₁, …, aₙ) ↦ \node(a₀, \unit{a₁} ∷ … ∷ \unit{aₙ} ∷ ε))
  \end{mathpar}
  \vspace{-0.15in}
  \Description{Hedgehog*'s Distribution Interpretation}
  \caption{\hedgehog*{}'s Distribution Interpretation.}
  \label{fig:distribution-interpretation*}
  \vspace{-0.1in}
\end{figure}

\ifextended
\begin{proposition}
  The interpretation $\Distrib$ is an "arrow".
\end{proposition}

\begin{proof}
  Follows from the fact that $\Dist$ is a "monad" and $(\blank \qfun \Tree{\blank})$ is an "arrow".
\end{proof}

\begin{lemma} \label[lemma]{lem:distrib-preservation}
  The following equations hold:
  \begin{align}
    \map(\curry(\arr_{\Sampling*}{f})){λ}
     & = \arr_{\Distrib}{f}
    \label{eq:dist-arr}                                          \\
    \map(\curry(\afirst_{\Sampling*}{f})){λ}
     & = \afirst_{\Distrib}(\map(\curry{f}{λ}))
    \label{eq:dist-first}                                        \\
    \map(\curry(\aleft_{\Sampling*}{f})){λ}
     & = \aleft_{\Distrib}{f}
    \label{eq:dist-left}                                         \\
    \map(\curry(f \acomp_{\Sampling*} g)){λ}
     & = \map(\curry{f}){λ} \acomp_{\Distrib} \map(\curry{g}){λ}
    \label{eq:dist-comp}                                         \\
    \map(\curry(\flip_{\Sampling*}{p})){λ}
     & = \flip_{\Distrib}{p}
    \label{eq:dist-flip}                                         \\
    \map(\curry{\shrink^n_{\Sampling*}}){λ}
     & = \shrink^n_{\Distrib}
    \label{eq:dist-shrink}
  \end{align}
\end{lemma}

\begin{proof} \phantom{A}
  \begin{enumerate}
    \item[\eqref{eq:dist-arr}] \begin{align*}
            \map(\curry(\arr_{\Sampling*}{f})){λ}
             & = \map(σ ↦ a ↦ \arr_{\Sampling*}{f}(σ, a)){λ} \\
             & = \map(σ ↦ a ↦ \unit(\app{f}{a})){λ}          \\
             & = \unit(a ↦ \unit(\app{f}{a}))                \\
             & = \arr_{\Distrib}{f}
          \end{align*}
    \item[\eqref{eq:dist-first}] \begin{align*}
            \map(\curry(\afirst_{\Sampling*}{f})){λ}
             & = \map(σ ↦ (a, c) ↦ \afirst_{\Sampling*}{f}(σ, (a, c))){λ}      \\
             & = \map(σ ↦ (a, c) ↦ \map(\blank, c)(\app{f}(σ, a))){λ}          \\
             & = \map(σ ↦ (a, c) ↦ \map(\blank, c)(\curry{f}{σ}{a})){λ}        \\
             & = \map(f' ↦ (a, c) ↦ \map(\blank, c)(f' a))(\map(\curry{f}){λ}) \\
             & = \afirst_{\Distrib}(\map(\curry{f}){λ})
          \end{align*}
    \item[\eqref{eq:dist-left}] \begin{align*}
            \map(\curry(\aleft_{\Sampling*}{f})){λ}
             & = \map(σ ↦ \elim{\map{\inj_1} ∘ \curry{f}{σ}, \unit ∘ \inj_2}){λ} \\
             & = \map(f' ↦ \elim{\map{\inj_1} ∘ f', \inj_2})(\map(\curry{f}){λ}) \\
             & = \aleft_{\Distrib}(\map(\curry{f}){λ})
          \end{align*}
    \item[\eqref{eq:dist-comp}] \begin{align*}
                  & \map(\curry(f \acomp_{\Sampling*} g)){λ}                                           \\
            {}={} & \map(σ ↦ a ↦ \bind(b ↦ \app{f}(\projr{σ}, b))(\app{g}(\projl{σ}, a))){λ}           \\
            {}={} & \bind(σ ↦ \unit(a ↦ \bind(b ↦ \app{f}(\projr{σ}, b))(\app{g}(\projl{σ}, a)))){λ}   \\
            {}={} & \bind(σ₁ ↦ \bind(σ₂ ↦ \unit(a ↦ \bind(b ↦ \app{f}(σ₂, b))(\app{g}(σ₁, a)))){λ}){λ} \\
            {}={} & \bind(σ₁ ↦ \map(σ₂ ↦ a ↦ \bind(b ↦ \app{f}(σ₂, b))(\app{g}(σ₁, a))){λ}){λ}         \\
            {}={} & \bind(σ₁ ↦ \map(σ₂ ↦ \bind(\curry{f}{σ₂}) ∘ \curry{g}{σ₁}){λ}){λ}                  \\
            {}={} & \bind(a₂ ↦ \map(σ ↦ \bind(\curry{f}{σ}) ∘ a₂){λ})(\map(\curry{g}){λ})              \\
            {}={} & \bind(a₂ ↦ \map(a₁ ↦ \bind{a₁} ∘ a₂)(\map(\curry{f}){λ}))(\map(\curry{g}){λ})      \\
            {}={} & \map(\curry{f}){λ} \acomp_{\Distrib} \map(\curry{g}){λ}
          \end{align*}
  \end{enumerate}
\end{proof}
\fi

\begin{theorem} \label[theorem]{thm:compositionality}
  \AP
  Suppose $\typed*[\sing]{t}{\GFun{\Unit}{\Bool}}$. Then $\dist*{t} = \tmsem*{t}^{\Distrib}$.
\end{theorem}

\ifextended
\begin{proof}
  Follows from \Cref{lem:distrib-preservation}.
\end{proof}
\fi

\AP
\Cref{thm:compositionality} states that, unlike \hedgehog{}'s "distribution semantics@@hedgehog", \hedgehog*{}'s "distribution semantics" \emph{are} "compositional", as witnessed by the "interpretation" $\Distrib$. Under $\Distrib$, "generator functions" are interpreted as distributions of functions into "shrink trees@@hedgehog". Note that, in \Cref{thm:compositionality}, $t$ is a "command term" and therefore represents an effectful computation.

\begin{theorem} \label[theorem]{thm:sound-transformations*}
  \Crefrange{eq:constant-folding}{eq:shrink-inlining} (\Cref{thm:sound-transformations,thm:unsound-transformations}) all hold for \hedgehog*{} terms under $\Distrib$, \emph{except} for \cref{eq:if-elim}.
\end{theorem}

\Cref{thm:sound-transformations*} demonstrates that \hedgehog*{} is considerably more suited for program optimization than \hedgehog{}. The only transformation which is not valid (as demonstrated in \Cref{ex:coin-semantics*}) is \cref{eq:if-elim}. However, we can recover a special case.

\begin{proposition} \label[proposition]{prop:sound-transformations*}
  The following equivalence holds:
  \begin{equation*}
    \semeq*{\Distrib}{\tlet{x}{\tflip{p}}{\tif{x}{t}{t}}}{t}
    \label{eq:if-elim'}
    \tag{\begin{NoHyper}\ref{eq:if-elim}$'$\end{NoHyper}}
  \end{equation*}
\end{proposition}

Intuitively, \Cref{prop:sound-transformations*} demonstrates that \cref{eq:if-elim} holds whenever the condition term does not shrink.

\subsection{From \hedgehog*{} to \hedgehog{}}

\begin{definition} \label[definition]{def:translation}
  The ""translation"" from \hedgehog*{} "types", "environments", and "terms" to \hedgehog{} "types@@hedgehog", "environments@@hedgehog", and "terms@@hedgehog" is given by the functions $\tytrans{\blank} : \Type* → \Type$, $\envtrans{\blank} : \Env* → \Env$, and $\tmtrans{\blank} : \Term* → \Term$ defined in \Cref{fig:translation}.
\end{definition}

\begin{figure}[t]
  \textbf{Types}
  \begin{mathpar}
    \tytrans{\Unit} = \Unit

    \tytrans{\Empty} = \Empty

    \tytrans{\Sum{τ₁}{τ₂}} = \Sum{\tytrans{τ₁}}{\tytrans{τ₂}}

    \tytrans{\Prod{τ₁}{τ₂}} = \Prod{\tytrans{τ₁}}{\tytrans{τ₂}}

    \tytrans{\Fun{τ₁}{τ₂}} = \Fun{\tytrans{τ₁}}{\tytrans{τ₂}}

    \tytrans{\GFun{τ₁}{τ₂}} = \Fun{\tytrans{τ₁}}{\Gen{\tytrans{τ₂}}}
  \end{mathpar}
  \textbf{Environments}
  \begin{mathpar}
    \envtrans{x₁ : τ₁, …, xₙ : τₙ} = x₁ : \tytrans{τ₁}, …, xₙ : \tytrans{τₙ}
  \end{mathpar}
  \textbf{Terms}
  \begin{mathpar}
    \tmtrans{x}
    = x

    \tmtrans{\tabsurd{t}}
    = \tabsurd{\tmtrans{t}}

    \tmtrans{\tunit}
    = \tunit

    \tmtrans{\tinl{t}}
    = \tinl{\tmtrans{t}}

    \tmtrans{\tinr{t}}
    = \tinr{\tmtrans{t}}

    \begin{array}{@{}l@{}}
      \tmtrans{\tmatch{t₁}{x₁}{t₂}{x₂}{t₃}} \\
      \quad = \tmatch*{\tmtrans{t₁}}{x₁}{\tlet{y}{\tmtrans{t₂}}{\treturn{y}}}{x₂}{\tlet{y}{\treturn{\tunit}}{\tmtrans{t₃}}}
    \end{array}

    \tmtrans{\tpair{t₁, t₂}}
    = \tpair{\tmtrans{t₁}, \tmtrans{t₂}}

    \tmtrans{\tfst{t}}
    = \tfst{\tmtrans{t}}

    \tmtrans{\tsnd{t}}
    = \tsnd{\tmtrans{t}}

    \tmtrans{\tlam{x}{τ}{t}}
    = \tlam{x}{\tytrans{τ}}{\tmtrans{t}}

    \tmtrans{\app{t₁}{t₂}}
    = \app{\tmtrans{t₁}}{\tmtrans{t₂}}

    \tmtrans{\treturn{t}}
    = \treturn{\tmtrans{t}}

    \tmtrans{\tlet{x}{t₁}{t₂}}
    = \tlet{x}{\tmtrans{t₁}}{\tmtrans{t₂}}

    \tmtrans{\tflip{p}}
    = \tflip{p}

    \tmtrans{\tshrink}
    = \tshrink
  \end{mathpar}
  \vspace{-0.2in}
  \Description{Translation from Hedgehog-Arrow to Hedgehog}
  \caption{Translation from \hedgehog*{} to \hedgehog{}.}
  \label{fig:translation}
  \vspace{-0.15in}
\end{figure}

\Cref{def:translation} defines our "translation" from \hedgehog*{} "terms" to \hedgehog{} "terms@@hedgehog". The only changes performed by $\tmtrans{\blank}$ are
\begin{enumerate*}
  \item replacing every instance of $\GFun{τ₁}{τ₂}$ with $\Fun{τ₁}{\Gen{τ₂}}$, and
  \item adding additional $\treturn$ computations to $\tmatchkeyword$ "terms".
\end{enumerate*}

\begin{proposition} \label[proposition]{prop:translation}
  The following statements hold:
  \begin{enumerate}
    \item Suppose $\typed*[Γ]{t}{τ}$. Then $\typed[\envtrans{Γ}]{\tmtrans{t}}{\tytrans{τ}}$.
    \item Suppose $\ctyped{Γ}{Δ}{t}{τ}$. Then $\typed[\envtrans{Γ}, \envtrans{Δ}]{\tmtrans{t}}{\tytrans{τ}}$.
  \end{enumerate}
\end{proposition}

\Cref{prop:translation} shows that our "translation" preserves "typing". As shown in \Cref{ex:coin-semantics*}, it is not the case that our "translation" preserves "sampling semantics". However, "distribution semantics" is preserved.

\begin{lemma} \label[lemma]{lem:substitution}
  Suppose $\semeq*{\Distrib}[\sing \semi \sing]{t₁}{t₂}[\Bool]$. Then $\dist{\tmtrans{t₁}} = \dist{\tmtrans{t₂}}$.
\end{lemma}

\subsection{Summary}

We have introduced \hedgehog*{}, a restricted variant of \hedgehog{} with a "compositional" "distribution semantics" (\Cref{thm:compositionality}). \hedgehog*{} achieves this by forcing different branches in a case analysis ($\tmatchkeyword$) to produce statistically independent values. This cannot be a simple semantic change due to the presence of higher-order computation, so \hedgehog*{} also places restrictions on (effectful) function application in the style of the arrow calculus \citep{arrow-calculus}. The resulting language admits nearly all the program transformations discussed in \Cref{thm:sound-transformations,thm:unsound-transformations}, except for one which holds in a restricted context. We provide a simple translation from \hedgehog*{} back to \hedgehog{}, which indicates that any \hedgehog{} program which corresponds to some \hedgehog*{} program is equally amenable to optimization.

\end{scope}

%% file: section/5-evaluation.tex
\section{Evaluation} \label{sec:evaluation}

In \Cref{sec:solution}, we described a modification of the \hedgehog{} language, \hedgehog*{}, which has a "compositional@@hedgehog*" "distribution semantics@@hedgehog*" which supports a much wider variety of program transformations than \hedgehog{} (\Cref{thm:sound-transformations*}). We now evaluate the effects of \hedgehog*{}'s design with respect to other practical considerations. Given that the \hedgehog*{} language obtains a "compositional@@hedgehog*" semantics by placing syntactic restrictions on the form of \hedgehog{} programs, the goal of our evaluation is to answer the following research questions:
\begin{itemize}
  \item[\textbf{RQ1:}] To what extent can a \hedgehog*{}-based language express existing generators?
  \item[\textbf{RQ2:}] What is the impact of \hedgehog*{}'s design on program size?
\end{itemize}

\subsection{Setup}

We have implemented \hedgehog*{} as a companion library to \hedgehog{}. The library, \texttt{hedgehog-arrow}, includes a type \texttt{Gen\,a\,b} of effectful functions from \texttt{a} to \texttt{b}. The Glasgow \haskell{} Compiler (GHC) includes syntactic extensions for arrows \citep{arrow-notation}, which allows programs to be written in a syntax similar to what is presented in \Cref{sec:solution}. For example, we translate \Cref{ex:coin-syntax*} to \haskell{} in \Cref{fig:coin-syntax-haskell}. Anonymous generator functions are constructed using the \texttt{proc} keyword, sequential composition is denoted using \texttt{do} notation, and generator function application is explicitly denoted using the operator ``$⤙$''.
\begin{figure}[t]
  \includegraphics[scale=0.8]{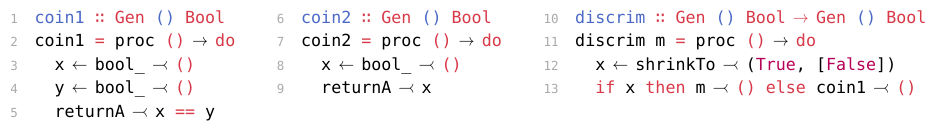}
  \vspace{-0.1in}
  \caption{\Cref{ex:coin-syntax*} translated to \haskell{}.}
  \Description{\Cref{ex:coin-syntax*} translated to \haskell{}}
  \label{fig:coin-syntax-haskell}
  \vspace{-0.1in}
\end{figure}

In addition to the two basic operators $\tflip$ and $\tshrink$ presented in \Cref{sec:analysis,sec:solution}, \texttt{hedgehog-arrow} also includes the following primitives:
\begin{enumerate}
  \item Operations \texttt{getSize $\dblcolon$ Gen a Size} and \texttt{resize $\dblcolon$ Gen a b $\rightarrow$ Gen (Size, a) b} for controlling the size of generated values.
  \item An operation \texttt{freeze $\dblcolon$ Gen a b $\rightarrow$ Gen a (Tree b)} which runs its argument, capturing all of its shrinking behaviour and saving it to a tree.
\end{enumerate}
The first feature comes from \quickcheck{} and solves certain termination issues \citep{quickcheck}. The second feature is used to modify or eliminate existing shrinking behaviour. For example, the function \texttt{prune}, which eliminates all shrinking behaviour from its argument, is implemented using \texttt{freeze} in \Cref{fig:prune}. Intuitively, \texttt{prune} behaves exactly like \texttt{freeze} but discards everything except the root value in the resulting tree (line 4).
\begin{figure}[t]
  \includegraphics[scale=0.8]{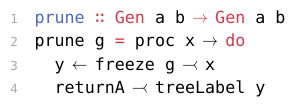}
  \vspace{-0.1in}
  \caption{The \texttt{prune} Function.}
  \Description{The \texttt{prune} Function}
  \label{fig:prune}
  \vspace{-0.1in}
\end{figure}

We perform our experiments on \hedgehog{}'s \texttt{Hedgehog.Gen} module. The \texttt{Hedgehog.Gen} module contains a large collection of common generators and generator combinators. We provide analogues for each of these generators and combinators in \texttt{hedgehog-arrow}. We also perform experiments on several modules from \hedgehog{}'s example repository\footnote{\href{https://github.com/hedgehogqa/haskell-hedgehog/tree/4e9045fafa9b50efdc68326cd9f643c0d8383a18}{Commit 4e9045fafa9b50efdc68326cd9f643c0d8383a18}}. We rewrote each of these modules to use \texttt{hedgehog-arrow}. Our full module set is given in \Cref{tab:examples}.

\begin{table}[t]
  \newcommand{\sizes}[2]{#1/#2 (\nmcEvaluate*{(#1-#2)/#2 * 100}[1]\%)}
  {\small
    \begin{tabular}{lll}
      \hline
      Module
       & \# Expressible
       & Module Size     \\
      \hline
      \texttt{Hedgehog.Gen}
       & 2/4/70 (76)
       & N/A             \\

      \texttt{Test.Example.Basic}
       & 0/5/13 (18)
       & \sizes{417}{393} \\

      \texttt{Test.Example.Confidence}
       & 0/0/1 (1)
       & \sizes{35}{37}           \\

      \texttt{Test.Example.Coverage}
       & 0/0/7 (7)
       & \sizes{251}{261}         \\

      \texttt{Test.Example.Exception}
       & 0/0/3 (3)
       & \sizes{82}{85}           \\

      \texttt{Test.Example.QuasiShow}
       & 0/0/2 (2)
       & \sizes{63}{60}           \\

      \texttt{Test.Example.Resource}
       & 0/0/3 (3)
       & \sizes{398}{401}         \\

      \texttt{Test.Example.RoundTrip}
       & 0/3/0 (6)
       & \sizes{225}{216}         \\

      \texttt{Test.Example.STLC}
       & 0/15/1 (16)
       & \sizes{425}{402}         \\
      \hline
    \end{tabular}}
  \captionsetup{justification=centering}
  \caption{
    Numbers of expressible generators per module and module sizes.
    Generator counts are reported as number of (not expressible/semi-expressible/fully-expressible (total)) generators.
    Module sizes are reported as (\hedgehog*{}/\hedgehog{}) module AST nodes.}
  \label{tab:examples}
  \vspace{-0.2in}
\end{table}

The source code for our implementation, examples, and experiments is available at \url{https://github.com/hedgehog-arr/hedgehog-arr}.

\subsection{Results}

\paragraph{\textbf{RQ1:} To what extent can a \hedgehog*{}-based language express existing generators?} \hedgehog*{} is intuitively a restricted version of \hedgehog{}, hence we investigate the impact this restriction has on expressivity in practice.

\Cref{tab:examples} lists, per module, how many generators and combinators within those modules
\begin{enumerate}
  \item cannot be expressed in \hedgehog*{},
  \item are \emph{semi-expressible}, i.e., can be expressed using non-expressible generators, and
  \item are expressible in \hedgehog*{}.
\end{enumerate}
We observed only two combinators in the first category: \texttt{choice} and \texttt{frequency}. To illustrate why these combinators are not expressible, we present a simplified implementation of \texttt{choice} in \Cref{fig:choice}. Intuitively, \texttt{choice xs} selects a random generator \texttt{x} from \texttt{xs} and executes it. During shrinking, \texttt{choice xs} may behave like any \texttt{x'} in \texttt{xs} which occurs earlier than \texttt{x}. Operationally, \texttt{choice xs} works by selecting a random integer \texttt{n} (line 4) and then executing the \texttt{n}th element of \texttt{xs} (line 5). The reason \texttt{choice} cannot be implemented is because it always passes the same sample to its selection \texttt{xs !! n}, which introduces a statistical dependency, which is explicitly disallowed by \hedgehog*{} (see \Cref{rem:restrictions}). The \texttt{frequency} combinator exhibits the same issue. We provide alternative implementations of \texttt{choice} and \texttt{frequency} in \texttt{hedgehog-arrow} which guarantee statistical independence during shrinking. Note that combinators are not generators themselves, but functions used to construct other generators: to our knowledge, no generator deliberately depends on the original behaviour of \texttt{choice} and \texttt{frequency}.
\begin{figure}[t]
  \includegraphics[scale=0.8]{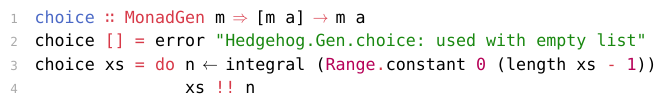}
  \vspace{-0.1in}
  \Description{The implementation of \texttt{choice} in \hedgehog{}}
  \caption{The implementation of \texttt{choice} in \hedgehog{}.}
  \label{fig:choice}
  \vspace{-0.1in}
\end{figure}

\noindent
\fbox{\parbox{\linewidth - 2\fboxsep}{
    \textbf{Summary:} \texttt{hedgehog-arrow} expresses almost all generators in our evaluation corpus, including those used in practical examples.}}

\paragraph{\textbf{RQ2:} What is the impact of \hedgehog*{}'s design on program size?} The restrictions \hedgehog*{} places on "command terms@@hedgehog*" may require users to employ verbose workarounds. We determine whether this is the case by measuring the impact of \hedgehog*{}'s design on program size.

\Cref{tab:examples} lists the size (in AST nodes) of both the \hedgehog{} and \texttt{hedgehog-arrow} version of each module. We do not compare the \texttt{Hedgehog.Gen} module because the \texttt{Hedgehog.Gen} primarily consists of re-exported internal code.

Our results show that using \texttt{hedgehog-arrow} has a negligible impact on program size.

\noindent
\fbox{\parbox{\linewidth - 2\fboxsep}{
    \textbf{Summary:} \texttt{hedgehog-arrow} programs are not significantly larger than their \hedgehog{} counterparts.}}

\subsection{Threats to Validity}

Our example set is entirely written in \haskell{}, which is a lazy and purely functional language. Our results may not generalize to strict, imperative languages. Our example set, while fairly large, does not include any particularly large examples (the largest being 122 lines), so our conclusions may not hold for larger programs.

%% file: section/6-related-work.tex
\section{Related Work} \label{sec:related-work}

\paragraph{Property-Based Testing and Shrinking} Property-based testing was popularized by the \quickcheck{} \citep{quickcheck} framework for Haskell, which features a probabilistic eDSL for specifying generators. In \quickcheck{}, a shrinker for a type a is specified separately as a function \texttt{shrinks $∷$ a $→$ [a]} from values to their immediate shrinks. In contrast, \hedgehog*{} employs \emph{generator-based shrinking}, where shrinkers are derived from annotated generator specifications.

Generator-based shrinking approaches were introduced by QuviQ \quickcheck{} \citep{quviq-quickcheck} for Erlang. QuviQ \quickcheck{} implements generator-based shrinking using  \emph{integrated shrinking}, where generators are represented as tree-valued random variables. Integrated shrinking is used by \hedgehog{} \citep{hedgehog} (and hence \hedgehog*{})  and many property-based testing frameworks in other languages such as RapidCheck~\citep{rapidcheck} for C++.

The \hypothesis{} \citep{hypothesis} library for Python implements a different approach to generator-based shrinking called \emph{internal shrinking}, where generators are represented as random variables drawing from a space of infinite bit sequences, and shrinks are obtained by re-executing generators on modified bit sequences obtained via a set of built-in heuristics. While hypothesis often produces better shrinks than approaches based on integrated shrinking for programs with no user annotations, hypothesis provides no mechanism for user control over shrinking. This is less than ideal when the built-in heuristics are insufficient. The \falsify{} library for Haskell \citep{falsify} rectifies this downside while retaining the advantages of internal shrinking by allowing users to specify custom shrinking behaviour. \Citet{falsify} identify the primary advantage of internal shrinking as \emph{hierarchical shrinking} behaviour, where shrinking is allowed to backtrack. The \hedgehog{} language actually exhibits hierarchical shrinking for generators implemented using \texttt{Applicative} functor operations \citep{applicative-functors}. Applicative functors and arrows are closely related \citep{applicative-functors}. Integrating hierarchical shrinking into \hedgehog*{} is left as future work.

Reflective generators \citep{reflecting} extend normal generators with \emph{partial monadic profunctor} operations. \Citet{reflecting} show how this extension lends itself to a number of features, including the ability to run a generator ``backwards'' to obtain a random seed from a user-provided test case. The reflective generator framework does not mandate a particular interpretation, but is incompatible with \hedgehog*{}'s semantics since \hedgehog*{}'s explicitly disallows the general (monadic) form of sequential composition that reflective generators support.

The formal semantics of generator-based shrinking languages have not been explored in detail prior to this work and hence these languages lack sound, formally-specified, compositional rules for reasoning about combined generator and shrinker behaviour.

  \paragraph{Optimizing Generator Execution Time.} The abstractions used by property-based testing libraries introduce a non-trivial amount of overhead. \citet{goldstein-2026} describe how to use staging techniques to eliminate this overhead. In another vein, QuickerCheck~\cite{claessen-2024} is an implementation of QuickCheck which achieves significant performance improvements by running test cases in parallel. Both works are orthogonal to ours: \citet{claessen-2024} and \citet{goldstein-2026} focus on optimizing the execution of generators, while our work enables optimizing generators themselves.

\paragraph{Semantics of Probabilistic Programming Languages} Quasi-Borel Spaces \citep{quasi-borel-spaces} form a semantic model of higher-order probabilistic programs, which we use to justify our various semantics. Quasi-Borel Pre-Domains \citep{quasi-borel-pre-domains} combine quasi-Borel spaces with $ω$-complete partial orders \citep{domain-theory} to allow modelling probabilistic programs with recursion. All of our results generalize to this setting except for (possibly) \Cref{thm:sampling-correct}. The proof of \Cref{thm:sampling-correct} relies on the fact that, whenever $\tmsem{t}^{\Sampling}{σ} = x$, then $\tmsem{t}^{\Sampling}{σ'} = x$ for all $σ'$ in some set of positive measure. This holds for all \emph{terminating} programs, but not necessarily for non-terminating programs. Investigating this case is left for future work. Note that a failure to generalize does not reduce the impact of \Cref{thm:sampling-correct} since it still applies to all terminating programs.

\paragraph{Semantics of Probabilistic and Nondeterministic Programming Languages} Shrinking can be viewed as a form of nondeterminism, and there is a large body of work addressing the combination of probabilistic and nondeterministic effects \citep{probability-and-nondeterminism-1,probability-and-nondeterminism-2,probability-and-nondeterminism-3,probability-and-nondeterminism-4,just-do-it}. Various authors \citep{probability-and-nondeterminism-1,probability-and-nondeterminism-2,just-do-it} observe that allowing probabilistic choice to distribute over nondeterministic choice, combined with the idempotence of probabilistic choice, leads to unexpected behaviour. \Cref{thm:sampling-correct} provides a similar type of observation, specific to \hedgehog{}. Placing this result in a more general setting is left to future work.

\paragraph{Arrows} Arrows are a generalization of monads \citep{monads} introduced to functional programming by \citet{arrows}. Arrows model languages with more restricted control flow than monads allow, which is crucial for our solution. Arrows alone do not model languages with case analysis: the \texttt{ArrowChoice} class \citep{arrows} in \haskell{} extends arrows with this ability, and corresponds to the $\aleft$ function required by \hedgehog*{} "interpretations@@hedgehog*". \hedgehog*{} is built on the \emph{arrow calculus} \citep{arrow-calculus}, which provides a more familiar presentation of arrows and their equational properties. \Citet{arrow-notation} describes a similar notation for \haskell{}, on which our implementation is based.

%% file: section/7-conclusion.tex
\section{Conclusion and Future Work} \label{sec:conclusion}

Existing languages that feature user-controlled generator-based shrinking lack sound, compositional semantics for reasoning about generator equivalence. In this paper, we showed that this problem is inherent to the design of existing languages, using \hedgehog{} as a canonical example. We defined \hedgehog*{}, a restricted version of \hedgehog{}, and showed that it has an appropriate compositional semantics. We also showed that \hedgehog*{} is expressive enough in practice by translating a large set of existing programs into \hedgehog*{}, and showed that \hedgehog*{} programs are not significantly more tedious to write than \hedgehog{} programs.

Generator-based shrinking is frequently much slower than manual shrinking~\cite{tweedale-2020,lyngle-2019}, which begs the question: why not use manual shrinking if performance is a concern? Anecdotally, we have found \falsify{} to be up to an order of magnitude faster than \hedgehog{} in some cases (although not quite as fast as \quickcheck{}), which indicates that poor performance with generator-based shrinking is largely a product of \hedgehog{}'s implementation strategy of representing generators as tree-valued random variables. In principle, generator-based shrinking should not induce additional overhead during the testing stage, as shrinking is not even performed on the vast majority of (successful) test cases.

Implementation strategies aside, our work demonstrates that generators written using Hedgehog-style generator-based shrinking frameworks are inherently \emph{less optimizable} than equivalent generators written in frameworks with manual shrinking, as many common optimizations are only allowed in the latter (\Cref{thm:unsound-transformations,thm:sampling-correct}). On the other hand, we have shown that many common optimizations are valid within \hedgehog*{} (\Cref{thm:sound-transformations*}). Implementing these optimizations (e.g. with rewrite rules \citep{ghc-rules} or a compiler plugin \citep{ghc-plugins} for Haskell programs) remain as future work.

As mentioned in \Cref{sec:related-work}, integrating \falsify{}-style \citep{falsify} hierarchical shrinking into \hedgehog*{} is also left as future work.